\documentclass[11pt]{article}
\usepackage[utf8]{inputenc}

\usepackage[margin=1.2in]{geometry}

\usepackage{natbib}
\usepackage{xcolor}
\usepackage{amsmath}
\usepackage{amssymb}
\usepackage{amsthm}
\usepackage{algorithmic}
\usepackage[]{algorithm}
\usepackage{hyperref}
\usepackage{cleveref}
\usepackage{graphicx}
\usepackage{thm-restate}
\usepackage[normalem]{ulem}
\usepackage{xspace}
\usepackage{arydshln}
\usepackage{wrapfig}
\usepackage{enumitem}
\usepackage{balance}
\usepackage{nicefrac}
\usepackage{url}
\usepackage{lmodern}
\setcitestyle{authoryear}

\newcommand{\silly}{\hspace*{1em}}
\newcommand{\Procedure}[2]{\STATE {{\bf procedure} {\sc #1}}} 
\newcommand{\State}{\STATE \silly}
\newcommand{\SubState}{\STATE \silly \silly}
\newcommand{\SubSubState}{\STATE \silly \silly \silly}
\newcommand{\Return}{{\bf return \/}}
\newcommand{\For}[1]{\State {\bf for\/} #1}
\newcommand{\SubFor}[1]{\SubState {\bf for\/} #1}
\newcommand{\If}[1]{\State {\bf if\/} #1}
\newcommand{\SubIf}[1]{\SubState {\bf if\/} #1}
\newcommand{\Else}{\State {\bf else\/}}
\newcommand{\SubElse}{\SubState {\bf else\/}}
\newcommand{\EndProcedure}{}
\newcommand{\EndIf}{}
\newcommand{\EndFor}{}

\newcommand{\nc}{\newcommand}
\newcommand{\DMO}{\DeclareMathOperator}

\newcommand{\remove}[1]{}

\allowdisplaybreaks

\newcommand{\poly}{\mathrm{poly}}
\nc{\MS}{\mathcal{S}}
\nc{\MP}{\mathcal{P}}
\nc{\MR}{\mathcal{R}}

\nc{\MZ}{\mathcal{Z}}
\nc{\cL}{\mathcal{L}}
\DMO{\Binom}{Binom}
\newcommand{\E}{\mathbb{E}}
\DMO{\Var}{Var}
\nc{\Geom}{\text{Geom}}

\newcommand{\tx}{\tilde{x}}
\newcommand{\talpha}{\tilde{\alpha}}

\newcommand{\tc}{\tilde{c}}

\newcommand{\bX}{\mathbf{X}}
\newcommand{\bC}{\mathbf{C}}

\newcommand{\tz}{\tilde{z}}

\newcommand{\tf}{\tilde{f}}

\newcommand{\teta}{\tilde{\eta}}
\newcommand{\tbeta}{\tilde{\beta}}

\newcommand{\tphi}{\tilde{\phi}}
\newcommand{\tv}{\tilde{v}}
\newcommand{\tn}{\tilde{n}}

\newcommand{\bx}{\mathbf{x}}

\newcommand{\bS}{\mathbf{S}}
\newcommand{\tbS}{\tilde{\bS}}

\newcommand{\bB}{\mathbb{B}}

\newcommand{\R}{\mathbb{R}}
\newcommand{\N}{\mathbb{N}}

\nc{\BN}{\mathbb{N}}

\nc{\BZ}{\mathbb{Z}}
\newcommand{\cR}{\mathcal{R}}

\newcommand{\cA}{\mathcal{A}}

\newcommand{\bzero}{\mathbf{0}}

\newcommand{\eps}{\varepsilon}
\nc{\esin}{\eps}

\newcommand{\cT}{\mathcal{T}}

\newcommand{\cX}{\mathcal{X}}
\newcommand{\cY}{\mathcal{Y}}

\DeclareMathOperator{\cost}{cost}

\DeclareMathOperator{\opt}{OPT}

\DeclareMathOperator{\SD}{SD}
\DeclareMathOperator{\hist}{hist}

\DeclareMathOperator{\level}{level}

\DeclareMathOperator{\leaves}{leaves}

\DeclareMathOperator{\children}{children}
\DeclareMathOperator{\cand}{cand}
\DeclareMathOperator{\enc}{Enc}
\DeclareMathOperator{\dec}{Dec}
\DeclareMathOperator{\bottom}{bottom}
\DeclareMathOperator{\polylog}{polylog}
\renewcommand{\varepsilon}{\epsilon}
\newcommand{\MT}{\mathrm{mt}}
\newcommand{\myvec}{\mathrm{vec}}

\newcommand{\algthreshold}{\textsc{ComputeThreshold}\xspace}
\newcommand{\buildtree}{\textsc{BuildTree}\xspace}

\newtheorem{theorem}{Theorem}
\newtheorem{proposition}[theorem]{Proposition}

\newtheorem{lemma}[theorem]{Lemma}

\newtheorem{definition}[theorem]{Definition}

\newtheorem{corollary}[theorem]{Corollary}

\newtheorem{fact}[theorem]{Fact}

\newcommand{\para}[1]{\smallskip \par \noindent {\bf #1.}}

\newcommand{\heps}{\nicefrac{\eps}{2}}
\newcommand{\hdel}{\nicefrac{\delta}{2}}
\newcommand{\rt}{\mathsf{t}}

\newcommand{\kmeans}{\textsf{\small
$k$-means}\xspace}
\newcommand{\kmedian}{\textsf{
$k$-median}\xspace}

\allowdisplaybreaks
\title{Locally Private $k$-Means in One Round}

\author{
Alisa Chang
\hspace*{1cm}
Badih Ghazi
\hspace*{1cm}
Ravi Kumar
\hspace*{1cm}
Pasin Manurangsi \\
Google, Mountain View, CA. \\
\texttt{\{alisac, pasin\}@google.com, \{badihghazi, ravi.k53\}@gmail.com}
}

\date{\today}

\begin{document}

\maketitle

\begin{abstract}

We provide an approximation algorithm for \kmeans clustering in the \emph{one-round} (aka \emph{non-interactive}) local model of differential privacy (DP).  This algorithm achieves an approximation ratio arbitrarily close to the best \emph{non private} approximation algorithm, improving upon previously known algorithms that only guarantee large (constant) approximation ratios.
Furthermore, this is the first constant-factor approximation algorithm for \kmeans that requires only \emph{one} round of communication in the local DP model, positively resolving an open question of~\citet{Stemmer20}. Our algorithmic framework is quite flexible; we demonstrate this by showing that it also yields a similar near-optimal approximation algorithm in the (one-round) shuffle DP model.
\end{abstract}


\section{Introduction}

The vast amounts of data collected for use by modern machine learning algorithms, along with increased public awareness for privacy risks, have stimulated intense research into privacy-preserving methods. Recently, differential privacy (DP) \citep{dwork2006calibrating,dwork2006our} has emerged as a popular definition, due to its mathematical rigor and strong guarantees. This has translated into several practical deployments within government agencies \citep{abowd2018us} and the industry \citep{erlingsson2014rappor,CNET2014Google, greenberg2016apple,dp2017learning, ding2017collecting}, and integration into widely-used libraries such as TensorFlow \citep{tf-privacy} and PyTorch \citep{pytorch-privacy}.

Among unsupervised machine learning tasks, clustering in general and \kmeans clustering in particular are one of the most basic and well-studied from both theoretical and practical points of view.  The research literature on \kmeans is extensive, with numerous applications not only in machine learning but also well  beyond computer science.  Computationally, finding optimal \kmeans cluster centers is NP-hard~\citep[e.g.,][]{aloise} and the best known polynomial-time algorithm achieves an approximation ratio of $6.358$~\citep{AhmadianNSW20}.

A natural question is to study clustering and \kmeans with DP guarantees. There is a growing literature on this topic~\citep{blum2005practical, nissim2007smooth, feldman2009private, GuptaLMRT10, mohan2012gupt, wang2015differentially, NissimSV16, nock2016k, su2016differentially, feldman2017coresets, BalcanDLMZ17, NissimS18, huang2018optimal,nock2016k,NissimS18, StemmerK18,Stemmer20, GKM20, jones2020differentially, chaturvedi2020differentially}. In contrast to non-private \kmeans, additive errors are necessary in this case, and a recent line of work has obtained increasingly tighter approximation ratios, culminating with the work of~\citet{GKM20}, which comes arbitrarily close to the smallest possible non-private approximation ratio (currently equal to $6.358$ as mentioned above). However, this last result only holds in the \emph{central} model of DP where a curator is entrusted with the raw data and is required to output private cluster centers.

The lack of trust in a central curator has led to extensive study of \emph{distributed} models of DP, the most prominent of which is the \emph{local} setting where the sequence of messages sent by each user is required to be DP. The state-of-the-art local DP \kmeans algorithm is due to~\citet{Stemmer20}; it achieves a constant (that is at least $100$ by our estimate) approximation ratio in a constant (at least five) number of rounds of interaction. In fact, all other (non-trivial) known local DP \kmeans algorithms are \emph{interactive}, meaning that the analyzer proceeds in multiple stages, each of which using the (private) outcomes from the previous stages.  This begs the question, left open by~\citet{Stemmer20}: is there a  non-interactive local DP approximation algorithm for \kmeans and can it attain the non-private approximation ratio?


\subsection{Our Contributions}

We resolve both open questions above by showing that there is a DP algorithm in the non-interactive local model with an approximation ratio that is arbitrarily close to that of any non-private algorithm, as stated below\footnote{An algorithm achieves a \emph{$(\kappa, t)$-approximation} if its output cost is no more than $\kappa$ times the optimal cost plus $t$, i.e., $\kappa$ is the approximation ratio and $t$ is the additive error (see \Cref{sec:kmeans-def}).}.

\begin{theorem} \label{thm:main-apx-local}
Suppose that there is a (not necessarily private) polynomial-time approximation algorithm for \kmeans with approximation ratio $\kappa$.
For any $\alpha > 0$, there is a one-round $\eps$-DP  $(\kappa(1 + \alpha), k^{O_{\alpha}(1)} \cdot \sqrt{nd} \cdot \polylog(n) / \eps)$-approximation algorithm for \kmeans in the local DP model. The algorithm runs in time $\poly(n, d, k)$.  
\end{theorem}

Here and throughout, we assume that the users and the analyzer have access to shared randomness. We further discuss the question of obtaining similar results using only private randomness in \Cref{sec:conc}. Furthermore, we assume throughout that $0 < \eps \leq O(1)$ and will not state this assumption explicitly in the formal statements.

The dependency on $n$ in the additive error above is also essentially tight, as~\citet{Stemmer20} proved a lower bound of $\Omega(\sqrt{n})$ on the additive error. This improves upon the protocol of \citet{Stemmer20}, where the additive error grows with $n^{1/2 + a}$ for an arbitrary small positive constant $a$.

In addition to the local model, our approach can be easily adapted to other distributed privacy models. To demonstrate this, we give an algorithm  in the shuffle model \citep{bittau17,erlingsson2019amplification,CheuSUZZ19}\footnote{The shuffle model of DP is a middle ground between the central and local models. Please see \Cref{app:shuffle-dp} for more details.} with a similar approximation ratio but with a smaller additive error of $\tilde{O}(k^{O_{\alpha}(1)} \cdot \sqrt{d}/ \eps)$.

\begin{theorem} \label{thm:main-apx-shuffle}
Suppose that there is a (not necessarily private) polynomial-time approximation algorithm for \kmeans with approximation ratio $\kappa$.
For any $\alpha > 0$, there is a one-round $(\eps, \delta)$-DP $(\kappa (1 + \alpha), k^{O_{\alpha}(1)} \cdot \sqrt{d} \cdot \polylog(nd/\delta) / \eps)$-approximation algorithm for \kmeans in the shuffle DP model.  The algorithm runs in time $\poly(n, d, k, \log(1/\delta))$.  
\end{theorem}

To the best of our knowledge, this is the first \kmeans algorithm in the one-round shuffle model with (non-trivial) theoretical guarantees. Furthermore, we remark that the above guarantees essentially match those of the central DP algorithm of~\citet{GKM20}, while our algorithm works even in the more restrictive shuffle DP model.

The main new ingredient in our framework is the use of a hierarchical data structure called a \emph{net tree}~\citep{Har-PeledM06} to construct a private \emph{coreset}\footnote{See \Cref{sec:kmeans-def} for the definition of coresets.} of the input points.  The decoder can then simply run any (non-private) approximation algorithm on this private coreset.  

The net tree~\citep{Har-PeledM06} is built as follows.  
Roughly speaking, each input point can be assigned to its ``representative'' leaf of a net tree; the deeper the tree is, the closer this representative is to the input point. To achieve  privacy, noise needs to be added to the number of nodes assigned to each leaf. This creates a natural tension when building such a tree: if the tree has few leaves, the representatives are far from the actual input points whereas if the tree has too many leaves, it suffers a larger error introduced by the noise. Our main technical contribution is an algorithm for constructing a net tree that balances between these two errors, ensuring that both do not affect the \kmeans objective too much. We stress that this algorithm works in the non-interactive local model, as each user can encode all possible representatives of their input point without requiring any interaction.

\para{Organization}
In Section~\ref{sec:prelims}, we provide  background on \kmeans and a description of the tools we need for our algorithm.  In Section~\ref{sec:nettrees}, we describe and analyze net trees, which are sufficient to solve private \kmeans in low dimensions.  In Section~\ref{sec:highdim}, we use dimensionality reduction to extend our results to the high-dimensional case.  In Section~\ref{sec:exp}, we present some empirical evaluation of our algorithm.  Section~\ref{sec:conc} contains the concluding remarks.

All the missing proofs and the results for the shuffle DP model are in the Appendix.

\section{Preliminaries}\label{sec:prelims}

We use $\|x\|$ to denote the Euclidean norm of a vector $x$. We use $\bB^d(x, r)$ to denote the closed radius-$r$ ball around $x$, i.e., $\bB^d(x, r) = \{y \in \R^d \mid \|x - y\| \leq r\}$.
Let $\bB^d = \bB^d(0, 1)$, the unit ball.
For every pair $T, T' \subseteq \R^d$ of sets, we use $d(T, T')$ as a shorthand for $\min_{x \in T, x' \in T'} \|x - x'\|$.

For every $m \in \N$, we use $[m]$ as a shorthand for $\{1, \dots, m\}$. 
For real numbers $a_1, \dots, a_t$, we use $\bottom_m(a_1, \dots, a_t)$ to denote the sum of the smallest $m$ numbers among $a_1, \dots, a_t$, where $t \geq m$.

We will be working with a generalization of multisets, where each element can have a weight associated with it.

\begin{definition}[Weighted Point Set]
A \emph{weighted point set} $\bS$ on a domain $D$ is a finite set of tuples $(x, w_{\bS}(x))$ where $x \in D$ and $w_{\bS}(x) \in \R_{\geq 0}$ denotes its weight; each $x$ should not be repeated in $\bS$.  It is useful to think of $\bS$ as the function $w_{\bS}: D \to \R_{\geq 0}$; this extends naturally (as a measure) to any subset $T \subseteq D$ for which we define $w_{\bS}(T) = \sum_{x \in T} w_{\bS}(x)$. We say that $w_{\bS}(D)$ is the \emph{total weight} of $\bS$ and we abbreviate it as $|\bS|$.
%
%
The \emph{average} of $\bS$ is defined as $\mu(\bS) := \sum_{x \in \bB^d} w_{\bS}(x) \cdot x / |\bS|$.
\end{definition}

Throughout the paper, we alternatively view a multiset $\bX$ as a weighted point set, where $w_{\bX}(x)$ denotes the (possibly fractional) number of times $x$ appears in $\bX$.

For an algorithm $\cA$, we use $\rt(\cA)$ to denote its running time.

\subsection{\kmeans and coresets}
\label{sec:kmeans-def}

Let $\bX \subseteq \bB^d$ be an input multiset and let $n = |\bX|$.
\begin{definition}[\kmeans]
The \emph{cost} of a set $\bC$ of centers with respect to an input multiset $\bX \subseteq \bB^d$ is defined as $\cost_{\bX}(\bC) := \sum_{x \in \bX} \left(\min_{c \in \bC} \|x - c\|\right)^2$. The \emph{\kmeans problem} asks, given $\bX$ and $k \in \N$, to find $\bC$ of size $k$ that minimizes $\cost_{\bX}(\bC)$; the minimum cost is denoted  $\opt_{\bX}^{k}$.

For a weighted point set $\bS$ on $\bB^d$, we define $\cost_\bS(C) := \sum_{x \in \bB^d} w_{\bS}(x) \cdot \left(\min_{c \in \bC} \|x - c\|\right)^2$, 
and let 
$\opt_{\bS}^{k} := \min_{|\bC| = k} \cost_{\bS}(\bC)$.

A set $\bC$ of size $k$ is a \emph{$(\kappa, t)$-approximate solution} for input $\bS$ if $\cost_{\bS}(\bC) \leq \kappa \cdot \opt^k_{\bS} + t$.
\end{definition}

It is also useful to define cost with respect to partitions, i.e., the mapping of points to the centers.

\begin{definition}[Partition Cost]
For a given partition $\phi: \bS \to [k]$ where $\bS$ is a weighted point set and an ordered set $\bC = (c_1, \dots, c_k)$ of $k$ candidate centers, we define the cost of $\phi$ with respect to $\bC$ as $\cost_{\bS}(\phi, \bC) := \sum_{x \in \bB^d} w_{\bS}(x) \cdot \|x - c_{\phi(x)}\|^2$.
Furthermore, we define the cost of $\phi$ as $\cost_{\bS}(\phi) := \min_{\bC \in (\R^d)^k} \cost_{\bS}(\phi, \bC)$.
\end{definition}

Note that a minimizer for the term $\min_{\bC \in (\R^d)^k} \cost_{\bS}(\phi, \bC)$ above is $c_i = \mu(\bS_i)$ where $\bS_i$ denotes the weighted point set corresponding to partition $i$, i.e., $w_{\bS_i}(x) = 
\begin{cases}
w_{\bS}(x) &\text{ if } \phi(x) = i, \\
0 &\text{otherwise.}
\end{cases}$

\para{Coresets}
In our algorithm, we will produce a weighted set that is a good ``approximation'' to the original input set. To quantify how good is the approximation, it is convenient to use the well-studied notion of coresets~\citep[see, e.g.,][]{Har-PeledM04}, which we recall below.

\begin{definition}[Coreset]
A weighted point set $\bS'$ is a \emph{$(k, \gamma, t)$-coreset} of a weighted point set $\bS$ if for every set $\bC \subseteq \bB^d$ of $k$ centers, it holds that $(1 - \gamma) \cdot \cost_{\bS}(\bC) - t \leq \cost_{\bS'}(\bC) \leq (1 + \gamma) \cdot \cost_{\bS}(\bC) + t$. When $k$ is clear from context, we refer to such an $\bS'$ as just a $(\gamma, t)$-coreset of $\bX$.
\end{definition}

\subsection{(Generalized) Monge’s Optimal Transport}

Another tool that will facilitate our proof is (a generalization of) optimal transport~\citep{monge1781memoire}. Roughly speaking, in Monge’s Optimal Transport problem, we are given two measures on a metric space and we would like to find a map that ``transports'' the first measure to the second measure, while minimizing the cost, which is often defined in terms of the total mass moved times some function of the distance. This problem can be ill-posed as such a mapping may not exist, e.g., when the total masses are different.  (We will often encounter this issue as we will be adding noise for differential privacy, resulting in unequal masses.)  To deal with this, we use a slightly generalized version of optimal transport where the unmatched mass is allowed but penalized by the $L_1$ difference,%
\footnote{Several works have defined similar but slightly different notions~\citep[see, e.g.,][]{benamou2003numerical,piccoli2014generalized}.}
defined next.

\begin{definition}[Generalized $L^2_2$ Transport Cost]
Let $\bS, \bS'$ be weighted point sets on $\bB^d$. Their \emph{generalized ($L_2^2$) Monge’s transport cost} of a mapping $\Psi: \bB^d \to \bB^d$ is defined as
\begin{align} \label{eq:transport-cost}
\MT(\Psi, \bS, \bS') & := 
\sum_{y \in \bB^d} w_{\bS}(y) \cdot \|\Psi(y) - y\|^2 
+ \sum_{x \in \bB^d} |w_{\bS}(\Psi^{-1}(x)) - w_{\bS'}(x)|.
\end{align}
Finally, we define the \emph{optimal generalized ($L_2^2$) Monge’s transport cost} from $\bS$ to $\bS'$ as
\begin{align} \label{eq:opt-transport-cost}
\MT(\bS, \bS') = \min_{\Psi: \bB^d \to \bB^d} \MT(\Psi, \bS, \bS').
\end{align}
\end{definition}

We remark that the minimizer $\Psi$ always exists because our weighted sets $\bS, \bS'$ have finite supports.  A crucial property of optimal transport we will use is that if the optimal transport cost between $\bS, \bS'$ is small relative to the optimal \kmeans objective, then $\bS'$ is a good coreset for $\bS$. This is encapsulated in the following lemma.
\begin{lemma} \label{lem:opt-transport-to-coreset}
For any $\xi \in (0, 1)$, any weighted point sets $\bS, \bS'$ over $\bB^d$, if $\MT(\bS, \bS') \leq \frac{\xi}{8(1 + 2/\xi)} \cdot \opt^k_{\bS} + t$ for some $t \geq 0$, then $\bS'$ is a $(\xi, 4(1 + 2/\xi)t)$-coreset of $\bS$.
\end{lemma}

We also use a slightly more refined bound stated below. It is worth to note that the first inequality below is very explicit as it also gives the mapping for $\bS$; we will indeed need such an additional property in our analysis when applying dimensionality reduction.

\begin{lemma} \label{lem:transport-cost-to-apx}
For any $\xi \in (0, 1]$, weighted point sets $\bS, \bS'$ over $\bB^d$, $\bC \in (\bB^d)^k$, $\phi : \bB^d \to [k]$ and $\Psi: \bB^d \to \bB^d$. We have
\begin{align*}
\cost_{\bS}(\phi \circ \Psi, \bC) \leq  (1 + \xi) \cdot \cost_{\bS'}(\phi, \bC) + 4(1 + 1/\xi) \cdot \MT(\Psi, \bS, \bS').
\end{align*}
and
\begin{align*}
\cost_{\bS'}(\bC) \leq (1 + \xi) \cdot \cost_{\bS}(\bC) + 4(1 + 1/\xi) \cdot \MT(\Psi, \bS, \bS').
\end{align*}
\end{lemma}

The proofs of both lemmas are deferred to \Cref{app:opt-transport-to-coreset-proof}

\subsection{Efficiently Decodable Nets}

Let $\cL \subseteq \bB^d$ be a finite set.  Its \emph{covering radius}, denoted $\rho(\cL)$, is defined as $\max_{x \in \bB^d} \min_{y \in \cL} \|x - y\|$. Its \emph{packing radius}, denoted $\gamma(\cL)$, is defined as the largest $\gamma$ such that the open balls around each point of $\cL$ of radius $\gamma$ are disjoint.  We say $\cL$ is a \emph{$(\rho, \gamma)$-net} if $\rho(\cL) \leq \rho$ and $\gamma(\cL) \geq \gamma$.

Using known results on lattices~\citep{rogers1959lattice,Micciancio04,MicciancioV13},~\citet{GKM20} show that there exist nets with $\gamma = \Omega(\rho)$ that are ``efficiently decodable'', i.e., given any point we can find points in $\cL$ close to it in $\exp(O(d))$ time.
\begin{lemma}[\citep{Micciancio04,GKM20}] \label{lem:decodable-nets}
For any given $\rho > 0$, there exists a $(\rho, \rho/3)$-net $\cL$ such that, for any given point $x \in \bB^d$ and any $r \geq \rho$, we can find all points in $\bB^d(x, r) \cap \cL$ in time $(1 + r/\rho)^{O(d)}$.
\end{lemma}

\subsection{(Generalized) Histogram and Vector Summation}
\label{subsec:histogram-prelim}

Finally, we need a few fundamental protocols as building blocks.  We first introduce a general aggregation problem.
\begin{definition}[Generalized Bucketized Vector Summation] 
\label{def:gbvs}
In the \emph{generalized bucketized vector summation problem}, each user holds a set $Y_i \subseteq Y$ of $T$ buckets and a vector $v_i \in \bB^d$. The goal is to determine, for a given $y \in Y$, the \emph{vector sum} of bucket $y$, which is $v_y := \sum_{i \in [n] \atop y \in Y_i} v_i$. An \emph{approximate generalized vector sum oracle} $\tv$ is said to be \emph{$\eta$-accurate} at $y$ if we have $\|v_y - \tv_y\| < \eta$. 
\end{definition}
The setting in this problem---each user holds a $d$-dimensional vector and contributes to $T$ buckets---generalizes many well-studied problems including (i) the \emph{histogram} problem where $T = d = 1$ and the (scalar) sum of a bucket is the \emph{frequency}, (ii) the \emph{generalized histogram} problem where $d = 1$ but $T$ can be more than one, and (iii) the \emph{bucketized vector summation} problem where $T = 1$ but $d$ can be more than one.

\subsection{Privacy Models}

We first recall the formal definition of differential privacy (DP).  For $\epsilon > 0$ and $\delta \in [0, 1]$, a randomized algorithm $\cA$ is \emph{$(\epsilon, \delta)$-DP} if for every pair $X, X'$ of inputs that differ on one point and for every subset $S$ of the algorithm's possible outputs, it holds that $\Pr[\cA(X) \in S] \leq e^{\epsilon} \cdot \Pr[\cA(X') \in S] + \delta$.  When $\delta = 0$, we simply say the algorithm is $\epsilon$-DP.

Let $n$ be the number of users, let $\bX = \{x_1, \ldots x_n\}$ and let the input $x_i$ be held by the $i$th user. An algorithm in the \emph{local DP model} consists of (i) an encoder whose input is the data held by one user and whose output is a sequence of messages
and (ii) a decoder, whose input is the concatenation of the messages from all the encoders and whose output is the output of the algorithm.  A pair $(\enc, \dec)$ is $(\eps, \delta)$-DP in the local model if for any input $\bX = (x_1, \dots , x_n)$, the algorithm $\cA(\bX) := (\enc(x_1), \ldots, \enc(x_n))$ is $(\epsilon, \delta)$-DP.  

For any generalized bucketized vector summation problem $\Pi$, we say that the pair $(\enc, \dec)$ is an
\emph{$(\eta, \beta)$-accurate} 
\emph{$(\eps, \delta)$-DP} algorithm (in a privacy model) if:
\begin{itemize}[nosep]
\item $\enc^{\Pi}_{(\eps, \delta)}$ is an $(\eps, \delta)$-DP algorithm that takes in the input and produces a randomized output,
\item $\dec^{\Pi}_{(\eps, \delta)}$ takes in the randomized output, a target bucket $y$ and produces an estimate vector sum for $y$,
\item For each $y \in Y$, the above oracle is $\eta$-accurate at $y$ with probability $1 - \beta$.
\end{itemize}

For the local DP model, the following generalized histogram guarantee is a simple consequence of the result of~\citet{BassilyNST20}:
\begin{theorem}[\citep{BassilyNST20}] \label{thm:hist-local}
There is an
$(O(\sqrt{n T^3 \log(|Y|/\beta)} / \eps), \beta)$-accurate $\eps$-DP algorithm for generalized histogram in the local model.  The encoder and the decoder run in time \\ $\poly(nT/\eps, \log |Y|)$.  
\end{theorem}
Using the technique of~\citet{BassilyNST20} together with that of~\citet{DuchiJW13}, we can obtain the following guarantee for generalized bucketized vector summation:
 
\begin{lemma} \label{lem:vec-sum-local}
There is an
$(O(\sqrt{n d T^3 \log(|Y|/\beta)} / \eps), \beta)$-accurate $\eps$-DP algorithm for generalized bucketized vector summation in the local model. The encoder and the decoder run in time $\poly(ndT/\eps, \log |Y|)$.  
\end{lemma}

The derivations of the above two bounds are explained in more detail in \Cref{app:hist-local-model}.

\section{Net Trees}
\label{sec:nettrees}

In this section, we describe \emph{net trees}~\citep{Har-PeledM06}, which are data structures that allow us to easily construct coresets of the inputs when the dimension is small. We remark that, although the main structure of net trees we use is similar to that of~\cite{Har-PeledM06}, there are several differences; the main one being the construction algorithm which in our case has to be done via (noisy) oracles, leading to considerable challenges.

\subsection{Description and Notation}

Let $\cL_1, \dots, \cL_T \subseteq \bB^d$ be a family of efficiently decodable nets, where $\cL_i$ has covering radius\footnote{The $2^i$ term here can be replaced by $\lambda^i$ for any $\lambda \in (0, 1)$; we use $2^i$ to avoid introducing yet another parameter.} $\rho_i := 1 / 2^i$ and packing radius $\gamma_i := \gamma / 2^i$. Furthermore, let $\cL_0 = \{\bzero\}, \rho_0 = \gamma_0 = 1$. For convenience, we assume that $\cL_0, \cL_1, \dots, \cL_T$ are disjoint; this is w.l.o.g. as we can always ``shift'' each net slightly so that their elements are distinct. 

For $i \in \{0, \dots, T\}$, let $\Psi_i: \bB^d \to \cL_i$ denote the map from any point to its closest point in $\cL_i$ (ties broken arbitrarily).

\para{Complete Net Tree}
Given a family of nets $\cL_1, \dots, \cL_T$, the 
\emph{complete net tree} is defined as a tree with $(T + 1)$ levels. For $i \in \{0, \dots, T\}$, the nodes in level $i$ are exactly the elements of $\cL_i$. Furthermore, for $i \in [T]$, the parent of node $z \in \cL_i$ is $\Psi_{i - 1}(z) \in \cL_{i - 1}$.  We use $\children(z)$ to denote the set of all children of $z$ in the complete net tree.

\para{Net Tree} 
A \emph{net tree} $\cT$ is a subtree of the complete net tree rooted at $\bzero$, where each node $z$ is either a leaf or all of $\children(z)$ must be present in the tree $\cT$. 

We will also use the following additional notation.  For each $i \in \{0, \dots, T\}$, let $\cT_i$ be the set of all nodes at level $i$ of tree $\cT$. Moreover, we use $\leaves(\cT)$ to denote the set of all leaves of $\cT$ and $\leaves(\cT_i)$ to denote the set of all leaves at level $i$. For a node $z \in \cT$, we use $\level(z) \in \{0, \dots, T\}$ to denote its level 
and $\children(z)$ to denote its children.

\para{Representatives}
Given a point $x \in \bB^d$, its \emph{potential representatives} are the $T + 1$ nodes $\Psi_T(x), \Psi_{T-1}(\Psi_T(x)), \dots, \Psi_0(\cdots(\Psi_{T}(x))\cdots)$
in the complete net tree. The \emph{representative} of $x$ in a net tree $\cT$, denoted by $\Psi_{\cT}(x)$, is the unique leaf of $\cT$ that is a potential representative of $x$. Note that $\Psi_{\cT}: \bB^d \to \cT$ induces a partition of points in $\bB^d$ based on the leaves representing them.

For a weighted point set $\bS$, its frequency at a leaf $z$ of $\cT$, denoted by $f_z$, is defined as the total weights of points in $\bS$ whose representative is $z$, i.e., $f_z = w_{\bS}(\Psi_{\cT}^{-1}(z))$.


\para{Representative Point Set} 
Let $\tf$ be a frequency oracle on domain $\cL_1 \cup \cdots \cup \cL_t$. The above partitioning scheme yields a simple way to construct a weighted point set from a net tree $\cT$. Specifically, the \emph{representative point set} of a tree $\cT$ (and frequency oracle $\tf$), denoted by $\bS_{\cT}$, is the weighted point set where each leaf $z \in \leaves(\cT)$ receives a weight of $\tf_{z}$. We stress that $\bS_{\cT}$ depends on the frequency oracle; however we discard it in the notation for readability.

\subsection{Basic Properties of Net Trees}

Before we proceed, we list a few important properties of net trees that both illustrate their usefulness and will guide our algorithms. The first property is that the potential representation of point $x$ at level $i$ cannot be too far from $x$:
\begin{lemma}[Distance Property] \label{lem:distance}
For any $x \in \bB^d$ and $i \in \{0, \dots, T\}$, we have $\|x - \Psi_i(\cdots(\Psi_{T}(x))\cdots)\| \leq 2^{1 - i}$.
\end{lemma}

\begin{proof}
Using the triangle inequality, we can bound
$\|x - \Psi_i(\cdots(\Psi_{T}(x))\cdots)\|$ above by  $\|\Psi_T(x) - x\| + \sum_{j=i}^{T-1} \|\Psi_j(\cdots(\Psi_{T}(x))\cdots) - \Psi_{j + 1}(\cdots(\Psi_{T}(x))\cdots)\|$
Since the covering radius of $\cL_j$ is $2^{-j}$, this latter term is at most $2^{-T} + \sum_{j=i}^{T - 1} 2^{-j} \leq 2^{1 - i}$ as desired.
\end{proof}
Second, we show that the number of children is small.
\begin{lemma}[Branching Factor] \label{lem:branching-factor}
For any $z \in \cL_0 \cup \cdots \cup \cL_{T - 1}$, we have $|\children(z)| \leq B := \left(1 + 2/\gamma\right)^d$.
\end{lemma}
\begin{proof}
Let $i = \level(z)$.
Since each $z' \in \children(z)$ has $z$ as its closest point in $\cL_i$, we have $\|z' - z\| \leq 2^{-i}$. Furthermore, since $\children(z) \subseteq \cL_{i + 1}$, $\children(z)$ form a $(\gamma \cdot 2^{-i-1})$-packing. As a result, a standard volume argument
implies that 
\[
|\children(z)| \leq 
\left(1 + \frac{2^{-i}}{\gamma \cdot 2^{-i-1}}\right)^d = \left(1 + 2/\gamma\right)^d.
\qedhere
\]
\end{proof}
We note that when applying dimensionality reduction in the next section, we will have\footnote{Since we will pick $\gamma > 0$ to be a constant, we hide its dependency in asymptotic notations throughout this section.} $d = O(\log k)$, meaning $B = k^{O(1)}$.

The last property
is an upper bound on the optimal transport cost from a given weighted point set to a representative point set created via a net tree $\cT$. Recall from \Cref{lem:opt-transport-to-coreset} that this implies a certain coreset guarantee for the constructed representative point set, a fact we will repeatedly use in the subsequent steps of the proof. 
The bound on the transport cost is stated in its general form below.

\begin{lemma} \label{lem:tree-transport}
For a weighted point set $\bS$ and a net tree $\cT$, let $f_z$ denote the frequency of $\bS$ on a leaf $z \in \leaves(\cT)$, i.e., $f_z = w_{\bS}(\Psi_{\cT}^{-1}(z))$. Let $S_{\cT}$ denote the representative point set constructed from $\cT$ and frequency oracle $\tf$. Then, 
\begin{align*}
\MT(\Psi_{\cT}, \bS, \bS_{\cT}) \leq 
\sum_{z \in \leaves(\cT)} \left( f_z \cdot (4\rho_{\level(z)}^2) + 
|f_z - \tf_z| \right).
\end{align*}
\end{lemma}

\begin{proof}
Recall by definition that $\MT(\Psi_{\cT}, \bS, \bS_{\cT})$ is equal to $\sum_{z \in \leaves(\cT)} |w_{\bS}(\Psi_{\cT}^{-1}(z)) - w_{\bS_{\cT}}(z)|+ \sum_{y \in \bB^d} w_{\bS}(y) \cdot \|\Psi_{\cT}(y) - y\|^2$. By construction of $\bS_{\cT}$, the first term is equal to $\sum_{z \in \leaves(\cT)} |f_z - \tf_z|$. Using \Cref{lem:distance}, the second term is bounded above by
\begin{align*}
\sum_{z \in \leaves(\cT)} \sum_{y \in \Psi_{\cT}^{-1}(z)} w_{\bS}(y) \cdot (2\rho_{\level(z)})^2 = \sum_{z \in \leaves(\cT)} f_z \cdot (4\rho_{\level(z)}^2).
\end{align*}
Combining the two bounds completes the proof.
\end{proof}

\subsection{Building the Net Tree}

Although we have defined net trees and shown several of their main properties, we have not addressed how a net tree should be constructed from a given (approximate) frequency oracle  (on $\cL_0 \cup \cdots \cup \cL_T$). \Cref{lem:tree-transport} captures the tension arising when constructing the tree: if we decide to include too many nodes in the tree, then all of the nodes contribute to the additive error, i.e., the second term in the bound; on other hand, if we include too few nodes, then many leaves will be at a small level, resulting in a large first term. In this section, we give a tree building algorithm that balances these two errors, and prove its guarantees.

We assume that each approximate frequency $\tf_{z}$ is non-negative.\footnote{This is w.l.o.g. as we can always take $\max(\tf_{z}, 0)$ instead.} 
The tree construction  (\Cref{alg:noisy-tree-construction}) itself is simple; we build the tree in a top down manner, starting with small levels and moving on to higher levels. At each level, we compute a \emph{threshold} $\tau$ on the number of nodes to expand.
We then only expand $\tau$ nodes with maximum approximate frequencies. The algorithm for computing this threshold $\tau$ (\algthreshold) and its properties will be stated next.

\begin{algorithm}[h]
\small
\caption{Building the Net Tree.} \label{alg:noisy-tree-construction}
\textbf{Oracle Access:} Frequency oracle 
$\tf$ on $\cL_0 \cup \cdots \cup \cL_T$ \\
\textbf{Parameters:} $k, a, \Gamma \in \N$

\begin{algorithmic}[1]
\Procedure{$\buildtree^{\tf}$}{}
\State $\cT \leftarrow$ root node $z = \bzero$ at level 0
\For{$i = 0, \dots, T - 1$}
\SubState $z^1_i, \dots, z^{m_i}_i \leftarrow$ level-$i$ nodes sorted in non-decreasing order of $\tf_{z}$
\SubState $\tau_i \leftarrow \algthreshold^{\tf}_{k, a, \Gamma}(z^1_i, \dots, z^{m_i}_i)$
\SubFor{$j = 0, \dots, \tau_i - 1$}
\SubSubState Add $\children(z^{m_i - j}_i)$ to $\cT$
\label{line:addchildren}
\EndFor
\EndFor
\State \Return $\cT$
\EndProcedure
\end{algorithmic}
\end{algorithm}

\para{Computing the Threshold}
Before we describe the threshold computation algorithm, we provide some intuition.  Recall \Cref{lem:tree-transport}, which gives an upper bound on the optimal transport cost and \Cref{lem:opt-transport-to-coreset}, which relates this quantity to the quality of the coreset. At a high level, these two lemmas allow us to stop branching as soon as the bound in \Cref{lem:tree-transport} becomes much smaller than $\opt^k_{\bS}$.  Of course, the glaring issue in doing so is that we do \emph{not} know $\opt^k_{\bS}$! It turns out however that we can give a \emph{lower bound} on $\opt^k_{\bS}$ based on the tree constructed so far. To state this lower bound formally, we first state a general lower bound on the \kmeans cost, without any relation to the tree.

\begin{lemma} \label{lem:opt-lb-generic-weighted}
Let $a, b, k \in \N$ and $r \in \R_{\geq 0}$.
Let $\bS$ be a weighted point set, and $T_1, \dots, T_{ka + b} \subseteq \R^d$ be any $ka + b$ disjoint sets such that for any point $c \in \R^d$ it holds that $|\{i \in [ka + b] \mid \bB^d(c, r) \cap T_i \ne \emptyset\}| \leq a$. Then, $\opt^k_{\bS} \geq r^2 \cdot \bottom_b(w_{\bS}(T_1), \ldots, w_{\bS}(T_{ka + b}))$.
\end{lemma}
\begin{proof}
Consider any set $\bC$ of $k$ candidate centers. From the assumption, there must be at least $b$ subsets $T_i$'s such that $d(\bC, T_i) \geq r$; for such a subset, its contribution to the \kmeans objective is at least $r^2 \cdot w_{\bS}(T_i)$. As a result, the total \kmeans objective is at least $r^2 \cdot \bottom_b(w_{\bS}(T_1), \dots, w_{\bS}(T_{ka + b}))$.
\end{proof}
This lets us prove a lower bound on the \kmeans objective for net trees:
\begin{corollary} \label{cor:opt-lb}
For any $\theta > 0$, let $r = \theta \cdot 2^{- i}$ and $a = \lceil \left(1 + (2 + \theta)/\gamma\right)^d \rceil$. Let 
$b \in \N$.
Suppose that there exist $(ka + b)$ level-$i$ nodes $\tz^1, \dots, \tz^{ka + b}$ in a net tree $\cT$. Furthermore, let $\bS$ be any multiset and $f$ the frequency of $\bS$. Then, we have $\opt^k_{\bS \cap \Psi_{\cT}^{-1}(\{\tz^1, \dots, \tz^{ka + b}\})} \geq r^2 \cdot \bottom_b\left(f_{\tz^1}, \dots, f_{\tz^{ka + b}}\right)$.
\end{corollary}
\begin{proof}
Consider any center $c \in \R^d$.
Recall from \Cref{lem:distance} that $\Psi_{\cT}^{-1}(\tz^j) \subseteq \bB^d(\tz^j, 2^{1 - i})$. In other words, if $\bB^d(c, r) \cap (S \cap \Psi_{\cT}^{-1}(\tz^j)) \ne \emptyset$, we must have 
\begin{align} \label{eq:hit-center-is-close}
\|c - \tz^j\| \leq r + 2^{1 - i} = (2 + \theta) 2^{-i},
\end{align}
implying that $\bB^d(\tz^j, \gamma \cdot 2^{-i}) \subseteq \bB^d(c, (2 + \theta) 2^{-i} + \gamma \cdot 2^{-i})$.

Furthermore, since $\tz^1, \dots, \tz^{ka + b} \subseteq \cL_i$ form a $(\gamma \cdot 2^{-i})$-packing, the balls $\bB^d(\tz^j, \gamma \cdot 2^{-i})$ are disjoint. This means that any point $c$ satisfies~\eqref{eq:hit-center-is-close} for at most 
\begin{align*}
\left(1 + \frac{(2 + \theta) \cdot 2^{-i}}{\gamma \cdot 2^{-i}}\right)^d \leq a.
\end{align*}
many $j$'s. Applying \Cref{lem:opt-lb-generic-weighted} completes the proof. 
\end{proof}
Thus, we can use $r^2 \cdot \bottom_b\left(f_{\tz^1}, \dots, f_{\tz^{ka + b}}\right)$ as a ``running lower bound'' on $\opt_{\bS}^k$. Still, we have to be careful as the additive error introduced will add up over all the levels of the tree; this can be an issue since we will select the number of levels to be as large as $\Theta(\log n)$. To overcome this, we make sure that the additive error introduced at each level is only ``charged'' to the optimum of the weighted point set corresponding to \emph{leaves} in that level. This ensures that there is no double counting in the error.

Below we formalize our cutoff threshold computation algorithm and prove its main property, which will be used later to provide the guarantees on the tree construction.

\begin{minipage}{\columnwidth}
\begin{algorithm}[H]
\small
\caption{Computing the Threshold.} \label{alg:cutoff-computation}
\textbf{Oracle Access: } Frequency oracle $\tf$ on $\cL_0 \cup \cdots \cup \cL_T$ \\
\textbf{Parameters:} $k, a, \Gamma \in \N$ \\
\textbf{Inputs:} Nodes $z^1, \dots, z^m$ from the same level of a net tree

\begin{algorithmic}[1]
\Procedure{$\algthreshold^{\tf}_{k, a, \Gamma}$}{$z^1, \dots, z^m$}
\For{$j \in [\min\{\Gamma, \lfloor m / ka\rfloor\}]$}
\SubIf{$\sum_{i=1}^{m-(j-1)ka} f_{z^i} \leq 2 \cdot \sum_{i=1}^{m - jka} f_{z^i}$} \label{step:check-threshold}
\SubSubState \Return $(j - 1)ka$ \label{step:return-threshold-smaller-than-opt}
\EndIf
\EndFor
\State \Return $\min\{m, \Gamma ka\}$ \label{step:return-threshold-}
\EndProcedure
\end{algorithmic}
\end{algorithm}
\end{minipage}

\begin{lemma} \label{lem:potential-func-noisy}
Suppose that $\tf$ is $\eta$-accurate at each of $z^1, \dots, z^m$ and suppose $\tf_{z^1} \leq \cdots \leq \tf_{z^m}$. Then, \textsc{\algthreshold}$^{\tf}_{a, \Gamma}(z^1, \dots, z^m)$ outputs $\tau$ satisfying
\[
\sum_{i=1}^{m - \tau} \tf_{z^i} \leq 2 \cdot 
\left(\sum_{i=1}^{m - \tau - ka} \tf_{z^i} \right)
+ \frac{n + m \eta}{2^\Gamma}.
\]
\end{lemma}

\begin{proof}
If the algorithm returns on line~\ref{step:return-threshold-smaller-than-opt}, then the inequality trivially holds due to the check on line~\ref{step:check-threshold} before. Furthermore, the inequality also trivially holds if $\tau = m$, as the left hand side is simply zero. As a result, we may assume that the algorithm returns $\tau = \Gamma k a < m$.
This means the condition on line~\ref{step:check-threshold} does \emph{not} hold for all $j \in [\Gamma]$. From this, we can conclude that
\begin{align*}
\tf_{z^1} + \cdots + \tf_{z^{m}} 
> 2\left(\tf_{z^1} + \cdots + \tf_{z^{m - ka}}\right)
>
\cdots
> 2^{\Gamma} \left(\tf_{z^1} + \cdots + \tf_{z^{m - \Gamma ka}}\right).
\end{align*}
Finally, since the frequency oracle is $\eta$-accurate at each of $z^1, \dots, z^m$, we have
\begin{align*}
\tf_{z^1} + \cdots + \tf_{z^{m}} \leq m\eta + f_{z^1} + \cdots + f_{z^m} \leq m\eta + n,
\end{align*}
where the latter inequality follows from the fact that each input point is mapped to at most one node at level $i$. Combining the above two inequalities yields the desired bound.
\end{proof}
We remark that $\tf_{z^1} + \cdots + \tf_{z^{m - \tau}}$ is the approximate frequency of points that will be mapped to the leaves at this level, which in turn gives the upper bound on the transport cost (in \Cref{lem:tree-transport}); whereas $\tf_{z^1} + \cdots + \tf_{z^{m - \tau - ka}}$ indeed governs the lower bound on the optimum \kmeans objective in \Cref{cor:opt-lb}. Intuitively, \Cref{lem:potential-func-noisy} thus allows us to ``charge'' the transport cost at each level to the lower bound on the optimum of the leaves at that level as desired. These arguments will be formalized in the next subsection.

\para{Putting it Together}
The main property of a net tree $\cT$ output by our tree construction algorithm is that its representative point set is a good coreset of the underlying input. This, plus some additional properties, is stated below.

\begin{theorem} \label{thm:net-tree-main-guarantee}
Let $\xi \in (0, 1)$.  Suppose that the frequency oracle $\tf$ is $\eta$-accurate on every element queried by the algorithm.
Let $\cT$ be the tree output by Algorithm~\ref{alg:noisy-tree-construction} where $\Gamma = \lceil \log n\rceil, T = \lceil 0.5 \log n \rceil, \theta = 8\sqrt{\frac{1 + 2/\xi}{\xi}}$, and $a$ be as in \Cref{cor:opt-lb}.  Let 
$N_T = 2^{O_{\xi}(d)} \cdot k \cdot (\log^2 n)$.  Then, 
\begin{itemize}[nosep]
\item The number of nodes in $\cT$ is $N_T$. Furthermore, this holds regardless of the frequency oracle accuracy.
\item $\MT(\Psi_{\cT}, \bS, \bS_{\cT}) \leq \frac{\xi}{8(1 + 2/\xi)} \cdot \opt^k_{\bS} + \eta \cdot O(N_T)$.
\item $S_{\cT}$ is a $\left(\xi, \eta \cdot O(N_T) 
\right)$-coreset of $S$.
\end{itemize}
Moreover, the tree construction algorithm runs in time $\poly(N_T)$ multiplied by the time to query $\tf$.
\end{theorem}

\begin{proof}
\begin{itemize}
\item First, from the operation of \Cref{alg:cutoff-computation}, we have $\tau_i \leq \Gamma ka \leq 2^{O_{\xi}(d)} \cdot k \cdot (\log n)$. 

By how $\cT$ is constructed, the number of internal nodes is $\tau_1 + \cdots + \tau_T$, which is at most $T\Gamma ka \leq 2^{O_{\xi}(d)} \cdot k \cdot (\log n)^2$. Finally, since by \Cref{lem:branching-factor} the branching factor $B$ is also just $2^{O(d)}$
, the total number of nodes is indeed $2^{O_{\xi}(d)} \cdot k \cdot (\log^2 n)$.

\item Using \Cref{lem:tree-transport}, we have
\begin{align}
\MT(\Psi_{\cT}, \bS, \bS_{\cT}) &\leq  \left(\sum_{z \in \leaves(\cT)} |f_z - \tf_z|\right) + \sum_{z \in \leaves(\cT)} f_z \cdot (4\rho_{\level(z)}^2) \nonumber \\
&\leq 4 \left(\sum_{z \in \leaves(\cT)} \tf_z \cdot (\rho_{\level(z)}^2)\right) + 2^{O_{\xi}(d)} \cdot k \cdot (\log^2 n) \cdot \eta, \label{eq:transport-cost-expand-first-step}
\end{align}
where the second inequality follows from the bound on the number of nodes in the first item and the $\eta$-accuracy guarantee of the frequency oracle.

To bound the summation term on the right hand side of~\eqref{eq:transport-cost-expand-first-step}, we may rearrange it as
\begin{align} \label{eq:expand-by-level}
\sum_{z \in \leaves(\cT)} \tf_z \cdot (\rho_{\level(z)}^2) &= \left(\sum_{i \in [T - 1]} 2^{-2i} \cdot \left(\sum_{z \in \leaves(\cT_i)} \tf_z\right) \right) + 2^{-2T} \cdot \sum_{z \in \leaves(\cT_T)} \tf_{z}.
\end{align}
Using \Cref{lem:potential-func-noisy}, we may bound the first term above by
\begin{align*}
&\left(\sum_{i \in [T - 1]} 2^{-2i} \cdot \left(\sum_{z \in \leaves(\cT_i)} \tf_z\right) \right) \\
&\leq \sum_{i \in [T - 1]} 2^{-2i} \cdot \left(2\bottom_{m_i - \tau_i - ka}\left((\tf_z)_{z \in \leaves(\cT_i})\right) + (n + |\cT_i|\eta) / 2^{\Gamma} \right) \\
&\leq \left(\sum_{i \in [T - 1]} 2^{1 - 2i} \cdot \bottom_{m_i - \tau_i - ka}\left((f_z)_{z \in \leaves(\cT_i})\right)\right) + O\left(1 + 2^{O_{\xi}(d)} \cdot k \cdot (\log^2 n)\right),
\end{align*}
where the second inequality follows from our choice of $\Gamma$ and the fact that $\eta \leq n$ which may be assumed without loss of generality (otherwise, we might just let the frequency oracle be zero everywhere) and the bound on the number of nodes in $\cT$ from the first item. Next, to bound the first summation term above, let $r_i := \theta \cdot 2^{-i}$. We have
\begin{align*}
&\left(\sum_{i \in [T - 1]} 2^{1 - 2i} \cdot \bottom_{m_i - \tau_i - ka}\left((f_z)_{z \in \leaves(\cT_i})\right)\right) \\
\text{(\Cref{cor:opt-lb})} &\leq \left(\sum_{i \in [T - 1]} 2^{1 - 2i}/r_i^2 \cdot \opt^k_{\bS \cap \Psi_{\cT}^{-1}(\leaves(\cT_i))}\right) \\
(\text{Our choice of } \theta) &= \frac{\xi}{32(1 + 2/\xi)} \cdot \left(\sum_{i \in [T - 1]} \opt^k_{\bS \cap \Psi_{\cT}^{-1}(\leaves(\cT_i))}\right) \\
&\leq \frac{\xi}{32(1 + 2/\xi)} \cdot \left(\opt^k_{\bigcup_{i \in [T - 1]} \left(\bS \cap \Psi_{\cT}^{-1}(\leaves(\cT_i))\right)}\right) \\
&\leq \frac{\xi}{32(1 + 2/\xi)} \cdot \opt^k_{\bS}.
\end{align*}
Finally, we may bound the second term in~\eqref{eq:expand-by-level} by
\begin{align*}
2^{-2T} \cdot \sum_{z \in \leaves(\cT_T)} \tf_{z} \leq (n + |\cT_T| \cdot \eta) / n \leq 2^{O_{\xi}(d)} \cdot k \cdot (\log^2 n),
\end{align*}
where we used the bound $2^{-2T} \leq 1/n$ which follows from our choice of $T$.

Combining the above four inequalities together, we get
\begin{align*}
\MT(\Psi_{\cT}, \bS, \bS_{\cT}) \leq \frac{\xi}{8(1 + 2/\xi)} \cdot \opt^k_{\bS} + 2^{O_{\xi}(d)} \cdot k \cdot (\log^2 n) \cdot \eta.
\end{align*}

\item Applying \Cref{lem:opt-transport-to-coreset} to the above inequality implies that $S_{\cT}$ is a $(\xi, 2^{O_{\xi}(d)} \cdot k \cdot (\log^2 n) \cdot \eta)$-coreset of $S$ as desired.
\end{itemize}

In terms of the running time, it is obvious that apart from Line~\ref{line:addchildren} in Algorithm~\ref{alg:noisy-tree-construction}, all other steps run in time $\poly(|\cT|)$, which is at most $\poly(N_T)$ times the running time of a frequency oracle call. As for Line~\ref{line:addchildren} in Algorithm~\ref{alg:noisy-tree-construction}, we may compute the set $\children(z)$ for some node $z \in \cL_j$ as follows. First, we use \Cref{lem:decodable-nets} to compute the set $\cand(z)$ of all nodes $z' \in \cL_{j + 1}$ such that $|z - z'| \leq \rho_j$; this takes time  $2^{O(d)}$. Next, for each $z' \in \cand(z)$, we check whether $z$ is its closest point in $\cL_j$, which once again can be done via \Cref{lem:decodable-nets} in time $2^{O(d)}$. Thus, each execution of Line~\ref{line:addchildren} in Algorithm~\ref{alg:noisy-tree-construction} takes only $2^{O(d)}$ time; hence, in total this step only takes $2^{O(d)} \cdot |\cT| \leq \poly(N_T)$ time.
\end{proof}

\section{From Low to High Dimension}
\label{sec:highdim}

The net tree-based algorithm 
can already be applied to give an approximation algorithm for \kmeans, albeit with an additive error of $2^{O_{\alpha}(d)} \cdot k \cdot (\log^2 n) \cdot \eta$. The exponential dependency on $d$ is  undesirable. In this section, we show how to eliminate this dependency via random projections, following the approach of~\citet{GKM20} for private clustering in the  central DP model. Specifically, the breakthrough work of~\citet{MakarychevMR19} allows one to randomly project to $d = O(\log k)$ dimensions while maintaining the objective for any given partition. 

\begin{theorem}[\citep{MakarychevMR19}] \label{thm:dim-red-cluster}
For every $0 < \tbeta, \talpha < 1$ and $k \in \N$, there exists $d' = O_{\talpha}\left(\log(k/\beta)\right)$ such that the following holds. Let $P$ be a random $d'$-dimensional subspace of $\R^d$ and $\Pi_P$ denote the projection from $\R^d$ to $P$. With probability $1 - \tbeta$, we have the following for all partitions $\phi: \bB^{d'} \to [k]$:
\[
\frac{1}{1 + \talpha} \leq
\frac{d \cdot \cost_{\Pi_P(\bS)}(\phi)}{d' \cdot \cost_{\bS}(\phi \circ \Pi_P)}
\leq 1 + \talpha.
\]
\end{theorem}

Our encoding algorithm first projects $x$ to $\tx$ in a given subspace $P$ and appropriately scales it to $x'$. It then computes all potential representatives (corresponding to the complete net tree of the nets $\cL_1, \dots, \cL_T$), and then encodes these representatives in the generalized histogram and the generalized bucketized vector summation encoders, the latter with the input vector $x$. This is presented more formally below in \Cref{alg:cluster-encoder}. (Note that we treat $x'_i$ as a vector in $\bB^{d'}$ directly; this is w.l.o.g. as we can rotate to make the basis of $P$ into the first $d'$ standard basis vectors.)

\begin{algorithm}[h!]
\small
\caption{Encoding Algorithm for \kmeans.}\label{alg:cluster-encoder}
\textbf{Input: } Point $x_i \in \bB^d$ of user $i$. \\
\textbf{Parameters: }  Privacy parameters $\eps, \delta$, nets $\cL_1, \dots, \cL_T$, $d'$-dimensional subspace $P$, and $\Lambda > 0$. \\
\textbf{Subroutines: } Encoders $\enc^{\hist}, \enc^{\myvec}$ for generalized histogram and bucketized vector summation.

\begin{algorithmic}[1]
\Procedure{kMeansEncoder$_{\eps, \delta, \Lambda, P, \cL_1, \dots, \cL_T}(x_i)$}{}
\State $\tx_i \leftarrow \Pi_P(x_i)$
\If{$\|\tx_i\| \leq 1 / \Lambda$}
\SubState $x'_i = \Lambda \tx$
\Else
\SubState $x'_i = \bzero$
\EndIf
\State $y^T_i \leftarrow $ Closest point to $x'_i$ in $\cL_T$ \label{line:cvp1}
\For{$j = T - 1, \dots, 1$}
\SubState $y^j_i \leftarrow$ Closest point to $y^{j + 1}_i$ in $\cL_j$ \label{line:cvp2}
\EndFor
\State $e^h_i \leftarrow \enc^{\hist}_{(\heps, \hdel)}(\{y^1_i, \dots, y^T_i\})$
\State $e^v_i \leftarrow \enc^{\myvec}_{(\heps, \hdel)}(\{y^1_i, \dots, y^T_i\}, x_i)$
\State \Return $(e^h_i, e^v_i)$
\EndProcedure
\end{algorithmic}
\end{algorithm}
 
To decode, we first use the encoded histogram to build a frequency oracle, from which we construct a net tree $\cT$ using the algorithm in \Cref{sec:nettrees}. We then run any approximation algorithm $\cA$ for \kmeans on the representative set of $\cT$. The output of $\cA$ gives a partition of the leaves of $\cT$ according to which centers are the closest. We then use the vector summation oracle on these partitions to determine the $k$ centers in the original (high-dimensional) space. A pseudo-code of this algorithm is given below as \Cref{alg:cluster-decoder}. We stress here that the approximation algorithm $\cA$ need \emph{not} be private.

\begin{algorithm}[h!]
\small
\caption{Decoding Algorithm for \kmeans.}\label{alg:cluster-decoder}
\textbf{Input: } Encoded inputs $e^h_1, e^v_1, \dots, e^h_n, e^v_n$. \\
\textbf{Parameters: }  Privacy parameters $\eps, \delta$, approximation algorithm $\cA$ for \kmeans. \\
\textbf{Subroutines: } Decoders $\dec^{\hist}, \dec^{\myvec}$ for generalized histogram and bucketized vector summation.

\begin{algorithmic}[1]
\Procedure{kMeansDecoder$_{\eps, \delta, \cA}(e^h_1, e^v_1, \dots, e^h_n, e^v_n)$}{}
\State $\tf \leftarrow$ frequency oracle from $\dec^{\hist}_{(\heps, \hdel)}(e^h_1, \dots, e^h_n)$
\State $\tv \leftarrow$ vector sum oracle from $\dec^{\myvec}_{(\heps, \hdel)}(e^v_1, \dots, e^v_n)$
\State $\cT \leftarrow \textsc{\buildtree}^{\tf}$
\State $\{c'_1, \dots, c'_k\} \leftarrow \cA(\bS_{\cT})$
\State $\phi_* \leftarrow$ mapping $\leaves(\cT) \rightarrow [k]$ where $\phi_*(z) = j$ iff $c'_j$ is closest to $z$ (with ties broken arbitrarily)
\For{$j = 1, \dots, k$}
\SubState $\tv^j \leftarrow \bzero$
\SubState $\tn^j \leftarrow 0$
\SubFor{$z \in \phi^{-1}_*(j)$}
\SubSubState $\tv^j \leftarrow \tv^j + \tv_z$
\SubSubState $\tn^j \leftarrow \tn^j + \tf_z$
\EndFor
\SubState $\tc^j = \tv^j / \max\{1, \tn^j\}$
\SubIf{$\|\tc^j\| \leq 1$}
\SubSubState $c_j \leftarrow \tc_j$
\SubElse
\SubSubState $c_j \leftarrow \tc_j / \|\tc_j\|$
\EndIf
\EndFor
\State \Return $\{c_1, \dots, c_k\}$
\EndProcedure
\end{algorithmic}
\end{algorithm}

A generic guarantee of our algorithm is stated next. As we will explain below, plugging known histogram/vector summation algorithms immediately yields our main results.

\begin{theorem} \label{thm:apx-main}
\textsc{kMeansEncoder}$_{\eps, \delta}$ is $(\eps, \delta)$-DP. Furthermore, suppose that the following hold:
\begin{itemize}[nosep]
\item $\cA$ is a $\kappa$-approximation algorithm for \kmeans.
\item $d'$ is as in \Cref{thm:dim-red-cluster} with $\tbeta = 0.1\beta, \talpha = 0.1\alpha$, and $\Lambda = \sqrt{\frac{0.01}{\log(n/\beta)} \cdot \frac{d}{d'}}$.
\item $P$ is a random $d'$-dimensional subspace of $\R^d$.
\item The parameters of \textsc{\buildtree} are as in Theorem~\ref{thm:net-tree-main-guarantee} with $\xi = 0.1\alpha$, and let $N_T = 2^{O_{\alpha}(d')} \cdot k \cdot (\log^2 n)$ be an upper bound on the number of nodes of $\cT$.
\item $(\enc^{\hist}_{(\heps, \hdel)}, \dec^{\hist}_{(\heps, \hdel)})$ is $(\eta, 0.1\beta / N_T)$-accurate for generalized histogram.
\item $(\enc^{\myvec}_{(\heps, \hdel)}, \dec^{\myvec}_{(\heps, \hdel)})$ is $(\teta, 0.1\beta / N_T)$-accurate for generalized bucketized vector summation.
\end{itemize}
Then, with probability $1 - \beta$, \textsc{kMeansDecoder} outputs a $\left(\kappa(1 + \alpha), k^{O_{\alpha}(1)} (\log^2 n) \left(\log(n/\beta) \cdot \eta + \teta\right)\right)$-approximate solution for \kmeans.  Moreover, the encoder runs in time $\poly(n d k^{O_\alpha(1)},$ $\rt(\enc^{\hist})$, $\rt(\enc^{\myvec}))$, and the decoder runs in time $\poly(n d k^{O_\alpha(1)}$, $\rt(\cA)$, $\rt(\dec^{\hist})$, $\rt(\dec^{\myvec}))$.
\end{theorem}

To prove this theorem, we will also use the following simple well-known fact~\citep[see e.g.,][Proposition 1]{AggarwalDK09}, which tell us an excess in \kmeans objective for each cluster in terms of the distance between the true center and the noised center. 

\begin{fact} \label{fact:cost-by-mean-error}
For any weighted point set $\bS$ and $c \in \R^d$, $\cost_{\bS}(c) - \opt_{\bS}^1 = |\bS| \cdot \|\mu(\bS) - c\|^2$.
\end{fact}

\begin{proof}[Proof of \Cref{thm:apx-main}]
Since each of $\enc^{\hist}_{(\heps, \hdel)}, \enc^{\myvec}_{(\heps, \hdel)}$ is $(\heps, \hdel)$-DP, basic composition theorem immediately implies that \textsc{ClusteringEncoder}$_{\eps, \delta}$ is $(\eps, \delta)$-DP.

Next, notice that we only call the oracles $\tf$ (resp. $\tv$) on the nodes of the tree $\cT$. Since the number of nodes is at most $N_{\cT}$, a union bound ensures that all of these queries provide $\eta$-accurate (resp. $\teta$-accurate) answers with probability at least $1 - 0.1\beta$. Henceforth, we may assume that such an accuracy guarantee holds for all queries.

For notational convenience, let $\bS = \{x_1, \dots, x_n\}, \bS' = \{x'_1, \dots, x'_n\}, \tbS = \{\tx_1, \dots, \tx_n\}, \bC = \{c_1, \dots, c_k\}$, and $\bC' = \{c'_1, \dots, c'_k\}$.

From \Cref{thm:net-tree-main-guarantee}, we have $\bS_{\cT}$ is a $\left(0.1\alpha, 2^{O_{\alpha}(d')} \cdot k \cdot (\log^2 n) \cdot \eta\right)$-coreset of $\bS'$. From this and from the $\kappa$-approximation guarantee of algorithm $\cA$, we have
\begin{align*}
\cost_{\bS_\cT}(\bC') \leq \kappa \cdot \opt^k_{\bS_T} \leq \kappa(1 + 0.1\alpha) \opt^k_{\bS} + \kappa \cdot 2^{O_{\alpha}(d')} \cdot k \cdot (\log^2 n) \cdot \eta.
\end{align*}

Let $\phi_*: \leaves(\cT) \to [k]$ denote the mapping from each leaf to its closest center, and let $\phi': \bS' \to [k]$ be $\phi' := \phi_* \circ \Psi_{\cT}$. With this notation, we have that the following holds with probability at least $1 - 0.1\beta$:
\begin{align*}
\cost_{\bS'}(\phi') &= \cost_{\bS'}(\phi_* \circ \Psi_{\cT}) \\
&\leq \cost_{\bS'}(\phi_* \circ \Psi_{\cT}, \bC') \\
(\text{\Cref{lem:transport-cost-to-apx}}) &\leq (1 + 0.1\alpha) \cdot \cost_{\bS_{\cT}}(\phi_*, \bC') + 4(1 + 1/\xi) \cdot \MT(\Psi_{\cT}, \bS', \bS_{\cT}) \\
(\text{\Cref{thm:net-tree-main-guarantee}}) &\leq (1 + 0.1\alpha) \cdot \cost_{\bS_{\cT}}(\phi_*, \bC') \\ &\qquad+ 4(1 + 1/\xi) \cdot \left(\frac{\xi}{8(1 + 2/\xi)} \cdot \opt^k_{\bS} + 2^{O_{\xi}(d)} \cdot k \cdot (\log^2 n) \cdot \eta\right) \\
(\text{From } \xi = 0.1\alpha) &\leq (1 + 0.15\alpha) \cdot \cost_{\bS_{\cT}}(\phi_*, \bC') + 2^{O_{\alpha}(d')} \cdot k \cdot (\log^2 n) \cdot \eta \\ 
(\text{guarantee of } \cA) &\leq \kappa(1 + 0.15\alpha) \cdot \opt^k_{\bS_{\cT}} + 2^{O_{\alpha}(d')} \cdot k \cdot (\log^2 n) \cdot \eta \\
(\text{\Cref{thm:net-tree-main-guarantee}}) &\leq \kappa(1 + 0.15\alpha)(1 + 0.1\alpha) \cdot \opt^k_{\bS'} + 2^{O_{\alpha}(d')} \cdot k \cdot (\log^2 n) \cdot \eta \\
&\leq \kappa(1 + 0.3\alpha) \cdot \opt^k_{\bS'} + 2^{O_{\alpha}(d')} \cdot k \cdot (\log^2 n) \cdot \eta.
\end{align*}

For notational convenience, we let $\tphi$ (respectively $\phi$) denote the canonical mapping from $\tbS$ (respectively $\bS$) to $[k]$ corresponding to $\phi'$. Standard concentration inequalities imply that with probability $1 - 0.1\beta$ we have $\|\tx_i\| \leq 1/\Lambda$, meaning that $x'_i = \Lambda \tx_i$ for all $i \in [n]$. When this holds, we simply have
\begin{align*}
\cost_{\bS'}(\phi') = \Lambda^2 \cdot \cost_{\tbS}(\tphi).
\end{align*}
Plugging this into the previous inequality, we have
\begin{align} \label{eq:projected-approx-guarantee}
\cost_{\tbS}(\tphi) &= \kappa(1 + 0.3\alpha) \cdot \opt^k_{\tbS} + (1/\Lambda^2) \cdot 2^{O_{\alpha}(d')} \cdot k \cdot (\log^2 n) \cdot \eta.
\end{align}

Next, applying Theorem~\ref{thm:dim-red-cluster}, the following holds with probability $1 - 0.1\beta$:
\begin{align*}
\cost_{\tbS}(\tphi) \geq (d/d') \cdot \cost_{\bS}(\phi) / (1 + 0.1\alpha),
\end{align*}
and
\begin{align*}
\opt^k_{\tbS} \leq (1 + 0.1\alpha)(d/d') \cdot \opt^k_{\bS}.
\end{align*}
Plugging the above two inequalities into~\eqref{eq:projected-approx-guarantee}, we arrive at
\begin{align}
\cost_{\bS}(\phi) &\leq \kappa(1 + 0.3\alpha)(1 + 0.1\alpha)^2 \opt^k_{\bS} + (d'/d) \cdot (1/\Lambda^2) \cdot 2^{O_{\alpha}(d')} \cdot k \cdot (\log^2 n) \cdot \eta \nonumber \\
&\leq \kappa(1 + \alpha) \opt^k_{\bS} + 2^{O_{\alpha}(d')} \cdot k \cdot (\log^2 n) \cdot \log(n/\beta) \cdot \eta. \label{eq:cost-partition-bound-main}
\end{align}

Next, we will bound $\cost_{\bS}(\bC)$ in comparison to $\cost_{\bS}(\phi)$ that we had already bounded above. To do this, first notice that
\begin{align}
\cost_{\bS}(\bC) \leq \cost_{\bS}(\phi, \bC) &= \sum_{j \in [k]} \cost_{\phi^{-1}(j)}(c_j) \nonumber \\
(\text{\Cref{fact:cost-by-mean-error}}) &= \sum_{j \in [k]} \left(\opt^1_{\phi^{-1}(j)} + |\phi^{-1}(j)| \cdot \|\mu(\phi^{-1}(j)) - c_j\|^2\right) \nonumber \\
&= \cost_{\bS}(\phi) + \left(\sum_{j \in [k]} |\phi^{-1}(j)| \cdot \|\mu(\phi^{-1}(j)) - c_j\|^2\right). \label{eq:k-means-bound-split}
\end{align}

Furthermore, since we assume that the vector summation oracle is $\teta$-accurate, using the triangle inequality we get
\begin{align} \label{eq:vector-sum-acc-bound}
\|v^j - \tv^j\| \leq \teta \cdot |\cT|.
\end{align}
Similarly, let $n^j = \phi^{-1}(j)$. From the fact that the frequency oracle is $\eta$-accurate, we have
\begin{align} \label{eq:count-sum-acc-bound}
|\tn^j - n^j| \leq \eta \cdot |\cT|.
\end{align}

We next consider two cases, based on how large $n^j$ is.
\begin{itemize}
\item Case I: $n^j \leq 2(\eta + \teta) \cdot |\cT|$. In this case, we simply use the straightforward fact that $\|\mu(\phi^{-1}(j)) - c_j\| \leq 2$. This gives
\begin{align*}
|\phi^{-1}(j)| \cdot \|\mu(\phi^{-1}(j)) - c_j\|^2 \leq 8(\eta + \teta) \cdot |\cT|.
\end{align*}
\item Case II: $n^j > 2(\eta + \teta) \cdot |\cT|$. In this case, first notice that 
\begin{align*}
\|\mu(\phi^{-1}(j)) - c_j\|
&\leq \|\mu(\phi^{-1}(j)) - \tc_j\| \\
&= \left\|\frac{v^j}{n^j} - \frac{\tv^j}{\tn^j}\right\| \\
&= \frac{1}{n^j \tn^j} \cdot \|v^j\tn^j - \tv^jn^j\| \\
(\text{From triangle inequality, \eqref{eq:count-sum-acc-bound}, } n^j > 2\eta \cdot |\cT|)  &\leq \frac{2}{(n^j)^2} \cdot \left(\|v^j(\tn^j - n^j)\| + \|(\tv^j - v^j)n^j\| \right) \\ 
(\text{From \eqref{eq:vector-sum-acc-bound}, \eqref{eq:count-sum-acc-bound}, and } \|v^j\| \leq n^j) &\leq \frac{2}{(n^j)^2} \cdot \left(n^j \cdot \eta \cdot |\cT| + \teta \cdot |\cT| \cdot n^j\right) \\
&= \frac{2(\eta + \teta)|\cT|}{n^j}.
\end{align*}
From this, we have
\begin{align*}
|\phi^{-1}(j)| \cdot \|\mu(\phi^{-1}(j)) - c_j\|^2 \leq n^j \cdot \frac{4(\eta + \teta)^2|\cT|^2}{(n^j)^2} \leq 2(\eta + \teta)|\cT|.
\end{align*}
\end{itemize}
Thus, in both cases, we have $|\phi^{-1}(j)| \cdot \|\mu(\phi^{-1}(j)) - c_j\|^2 \leq 8(\eta + \teta)|\cT|$. Plugging this back to \eqref{eq:k-means-bound-split}, we get
\begin{align*}
\cost_{\bS}(\bC) \leq \cost_{\bS}(\phi) + 8k(\eta + \teta)|\cT| \leq \cost_{\bS}(\phi) + 2^{O_{\alpha}(d')} \cdot k^2 \cdot (\log^2 n) \cdot (\eta + \teta),
\end{align*}
where the second inequality follows from the first item of \Cref{thm:net-tree-main-guarantee}. Finally, plugging this into~\eqref{eq:cost-partition-bound-main}, we can conclude that
\begin{align*}
\cost_{\bS}(\bC) 
&\leq \kappa(1 + \alpha) \opt^k_{\bS} + 2^{O_{\alpha}(d')} \cdot k \cdot (\log^2 n) \cdot \log(n/\beta) \cdot \eta \\ &\qquad + 2^{O_{\alpha}(d')} \cdot k^2 \cdot (\log^2 n) \cdot (\eta + \teta) \\
(\text{From } d = O_{\alpha}(\log k)) &\leq \kappa(1 + \alpha) \opt^k_{\bS} + k^{O_{\alpha}(1)} (\log^2 n) \left(\log(n/\beta) \cdot \eta + \teta\right),
\end{align*}
as desired.

The running time claim for the decoder (Algorithm~\ref{alg:cluster-decoder}) follows immediately from the running time of \textsc{BuildTree} from \Cref{thm:net-tree-main-guarantee} and from $2^{O(d')} = k^{O_{\alpha}(1)}$ due to our choice of $d'$. As for the encoder (Algorithm~\ref{alg:cluster-encoder}), it is clear that every step runs in $\poly(n d k,$ $\rt(\enc^{\hist})$, $\rt(\enc^{\myvec}))$ time, except Lines~\ref{line:cvp1} and~\ref{line:cvp2} where we need to find a closest point in $\cL_j$ from some given point. However, \Cref{lem:decodable-nets} ensures that this can be computed in time $2^{O(d')}$ which is equal to $k^{O_{\alpha}(1)}$ due to our choice of parameters. This completes our proof.
\end{proof}
\Cref{thm:apx-main} allows us to easily derive approximation algorithms for \kmeans in different distributed models of DP, by simply plugging in the known generalized histogram/vector summation guarantees.

\subsection{Approximation Algorithm for Local DP}

Next we consider the local DP model and prove \Cref{thm:main-apx-local}.

\begin{proof}[Proof of \Cref{thm:main-apx-local}]
Let $\beta = 0.1$. From \Cref{thm:hist-local}, there is an $(\eta, 0.1\beta / N_T)$-accurate $0.5\eps$-DP algorithm for generalized histogram with
\[
\eta = O(\sqrt{n T^3 \log(N_T |\cL_1 \cup \cdots \cup \cL_T| / \beta)} / \eps).
\]
Since we set $T = O(\log n)$ (in \Cref{thm:net-tree-main-guarantee}), $N_T = k^{O_\alpha(1)} \cdot \poly\log n$ by our choice of parameters, and since $|\cL_1 \cup \cdots \cup \cL_T| \leq \exp(O(Td'))$ by a volume argument, we get $\eta = O(\sqrt{n} \poly\log(n) / \eps)$.

Similarly, from \Cref{lem:vec-sum-local}, there is an $(\teta, 0.1\beta / N_T)$-accurate $0.5\eps$-DP algorithm for generalized histogram with
\[
\teta = O(\sqrt{nd T^3 \log(d N_T |\cL_1 \cup \cdots \cup \cL_T| / \beta)} / \eps),
\]
which as before yields $\teta = O(\sqrt{nd} \cdot \polylog(n) / \eps)$.  Plugging this into \Cref{thm:apx-main}, we indeed arrive at a one-round $\eps$-local DP $(\kappa (1 + \alpha), k^{O_{\alpha}(1)} \cdot \sqrt{nd} \cdot \polylog(n) / \eps)$-approximation algorithm for \kmeans (with failure probability $0.1$).  It is easy to verify that the encoder and the decoder run in time $\poly(n, d, k^{O_\alpha(1)})$.
\end{proof}

\section{Experiments}\label{sec:exp}

\textbf{Implementation Details.} 
We modify our algorithm in several places to make it more practical. First, instead of using nets we use locality-sensitive hashing (LSH). Specifically, given LSH $g_1, \dots, g_T$, the level-$i$ representation of $x$ is now $z_i = (g_1(x), \dots, g_T(x))$. In this sense, our tree bears a strong resemblance to the so-called \emph{LSH forests}~\citep{BawaCG05,AndoniRN17}. As for the specific family of hashes, we use SimHash~\citep{Charikar02} in which a random vector $v_i$ is picked and $g_i(x)$ is the sign of $\left<v_i, x\right>$. Since LSH can already be viewed as a dimensionality reduction method, we skip the random projection at the beginning of the algorithm.  Consequently, we also directly compute the approximate centers of all the nodes in the tree and then use a non-private algorithm (in particular, $k$-means$++$\citep{arthur2006k}) to compute the $k$ centers on this privatized dataset. 
Details on the choice of parameters can be found in~\Cref{app:exp}.

\textbf{Dataset Generation.}
We use mixtures of Gaussians, which are generated as follows. For a separation parameter $r$ and the number $k$ of clusters, first pick $k$ centers uniformly at random from the sphere of radius slightly less than one (i.e., $1 - \Theta(r)$). Then, for each center, we create $n/k$ points by adding to the center a Gaussian-distributed vector whose expected norm is $1/r$. Finally, we project any point that is outside of the unit ball back into it. Note that we run our algorithm on this dataset using the same value of $k$.

Although these datasets are relatively easy, we would like to stress that we are not aware of any prior experiments or implementations of non-interative local DP algorithms.%
\footnote{Very recently, \citep{XiaHuaTongZhong} have reported  experimental results for \kmeans in the local DP model. However, their algorithm,  which is based on Lloyd's iteration, requires interaction.}\footnote{We remark that we also tried some ``naive'' baseline algorithms, such as noising each point using Gaussian/Laplace noise and running $k$-means++ on this noisy dataset. However, they do not produce any meaningful clustering even for $n$ as large as $10^6$.}
Furthermore, we point out that even experiments in the central model of~\citet{BalcanDLMZ17}, which uses almost the same data generation process (including separation parameters) as ours, suggest that the datasets are already challenging for the much more relaxed central DP model.

\textbf{Results Summary.} \Cref{fig:exp_main} presents the \emph{normalized} \kmeans objective (i.e., the \kmeans objective divided by $n$) as $n, \eps$ or $k$ vary. Note that the ``trivial'' clustering where we simply use the origin as the center has normalized objective roughly equal to one. Our clustering algorithm significantly improves upon these when $n$ is sufficiently large. More importantly, as either $n$ or $\eps$ increases, the normalized objective decreases, as predicted by theory. Finally, when the number of clusters $k$ becomes larger, our algorithm suffers larger errors, once again agreeing with theory. 

\begin{figure*}[h]
\centering
\includegraphics[width=0.32\textwidth]{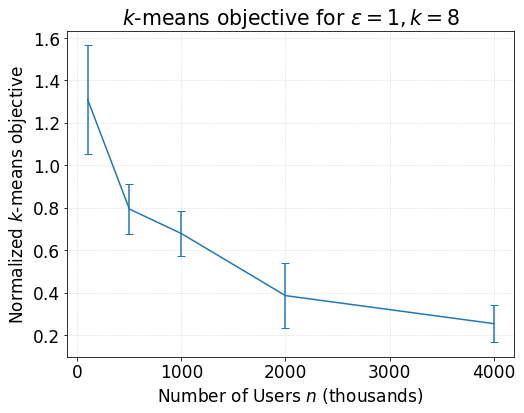}
\includegraphics[width=0.32\textwidth]{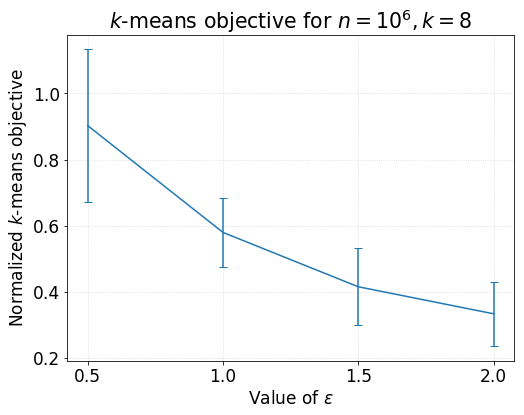}
\includegraphics[width=0.32\textwidth]{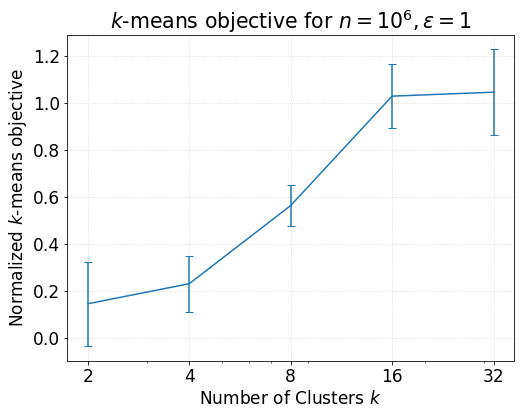}
\caption{Normalized \kmeans objective of the output clusters for varying $n, \eps$ or $k$ for $d = 100, r = 100$. Each set of parameters is run 10 times; the average and the standard deviation of the normalized \kmeans objectives are included.}
\label{fig:exp_main}
\end{figure*}

\section{Conclusions and Open Questions}
\label{sec:conc}

We give private approximation algorithms for \kmeans clustering whose ratios are essentially the same as those of non-private algorithms in both the (one-round) local DP and the (one-round) shuffle DP models. An interesting open question is to extend this result to other clustering objectives, such as \kmedian. While the net trees can also be applied to \kmedian with little change, the techniques we use to handle high dimensions do not directly carry over. This is due to the fact that, in \kmeans, it is simple to find a center of a cluster in one round, by finding the average of all the points. However, finding a cluster center in \kmedian (to the best of our knowledge) requires solving a linear program and it seems challenging to do so non-interactively.

Even for the \kmeans problem itself, there are still several interesting questions. First, as mentioned in the Introduction, our algorithm relies crucially on public randomness in the dimensionality reduction step, where every user has to project onto the same random subspace $P$. Can \kmeans be approximated almost optimally via a one-round local DP algorithm that uses only private randomness? 

Another direction is to tighten the additive error. In the central model, this has recently been investigated in~\citet{chaturvedi2020differentially,jones2020differentially}, who achieved essentially tight additive errors in terms of $k$ (and in certain regimes $d, n$). It is a natural question to determine such a tight dependency in the (non-interactive) local model too.

\balance
\bibliographystyle{abbrvnat}
\bibliography{refs}

\setlength{\abovedisplayskip}{8pt}
\setlength{\belowdisplayskip}{8pt}

\appendix

\section{Proofs of \Cref{lem:opt-transport-to-coreset} and \Cref{lem:transport-cost-to-apx}}
\label{app:opt-transport-to-coreset-proof}

To prove \Cref{lem:transport-cost-to-apx}, we will use the following fact, which is simple to verify.

\begin{fact}  \label{fact:p-pow-ineq}
For any $\xi > 0$ and any $a, b \in \R$, $(a + b)^2 \leq (1 + \xi) a^2 + (1 + 1/\xi) \cdot b^2$.
\end{fact}

\begin{proof}[Proof of \Cref{lem:transport-cost-to-apx}]
We have
\begin{align*}
\cost_{\bS}(\phi \circ \Psi, \bC) &= \sum_{y \in \bB^d} w_{\bS}(y) \cdot \|y - c_{\phi(\Psi(y))}\|^2 \\
&= \sum_{x \in \bB^d} \sum_{y \in \Psi^{-1}(x)} w_{\bS}(y) \cdot \|y - c_{\phi(x)}\|^2 \\
(\text{Triangle Inequality}) &\leq \sum_{x \in \bB^d} \sum_{y \in \Psi^{-1}(x)} w_{\bS}(y) \cdot \left(\|x - c_{\phi(x)}\| + \|x - y\|\right)^2 \\
(\text{\Cref{fact:p-pow-ineq}}) &\leq \sum_{x \in \bB^d} \sum_{y \in \Psi^{-1}(x)} w_{\bS}(y) \cdot \left((1 + \xi) \cdot \|x - c_{\phi(x)}\|^2 + (1 + 1/\xi) \cdot \|x - y\|^2\right) \\
&= (1 + \xi) \cdot \left(\sum_{x \in \bB^d} w_{\bS}(\Psi^{-1}(x)) \cdot \|x - c_{\phi(x)}\|^2\right) \\ &\qquad+ (1 + 1/\xi) \cdot \left(\sum_{y \in \bB^d} w_{\bS}(y) \cdot \|\Psi(y) - y\|^2\right) \\
&\leq (1 + \xi) \cdot \left(\sum_{x \in \bB^d} w_{\bS'}(x) \cdot \|x - c_{\phi(x)}\|^2\right) \\ &\qquad+ 4(1 + \xi) \cdot \left(\sum_{x \in \bB^d} |w_{\bS}(\Psi^{-1}(x)) - w_{\bS'}(x)|\right) \\ &\qquad+ (1 + 1/\xi) \cdot \left(\sum_{y \in \bB^d} w_{\bS}(y) \cdot \|\Psi(y) - y\|^2\right) \\
&\leq (1 + \xi) \cdot \cost_{\bS'}(\phi, \bC) + 4(1 + 1/\xi) \cdot \MT(\Psi, \bS, \bS'),
\end{align*}
yielding the first inequality.

To prove the second inequality, let $\phi_*: \bB^d \to [k]$ denote the map of each point to its closest center in $\bC$. We have
\begin{align*}
\cost_{\bS'}(\bC) &= \sum_{x \in \bB^d} w_{\bS'}(x) \cdot \|x - c_{\phi_*(x)}\|^2 \\
&\leq 4 \left(\sum_{x \in \bB^d} |w_{\bS}(\Psi^{-1}(x)) - w_{\bS'}(x)|\right) + \sum_{x \in \bB^d} w_{\bS}(\Psi^{-1}(x)) \cdot \|x - c_{\phi_*(x)}\|^2 \\
&= 4 \left(\sum_{x \in \bB^d} |w_{\bS}(\Psi^{-1}(x)) - w_{\bS'}(x)|\right) + \sum_{x \in \bB^d} \sum_{y \in \Psi^{-1}(x)} w_{\bS}(y) \cdot \|x - c_{\phi_*(x)}\|^2 \\
&\leq 4 \left(\sum_{x \in \bB^d} |w_{\bS}(\Psi^{-1}(x)) - w_{\bS'}(x)|\right) + \sum_{x \in \bB^d} \sum_{y \in \Psi^{-1}(x)} w_{\bS}(y) \cdot \|x - c_{\phi_*(y)}\|^2 \\ 
(\text{Triangle Inequality}) &\leq 4 \left(\sum_{x \in \bB^d} |w_{\bS}(\Psi^{-1}(x)) - w_{\bS'}(x)|\right) \\ &\qquad+ \sum_{x \in \bB^d} \sum_{y \in \Psi^{-1}(x)} w_{\bS}(y) \cdot (\|y - c_{\phi_*(y)}\| + \|x - y\|)^2 \\
&\leq 4 \left(\sum_{x \in \bB^d} |w_{\bS}(\Psi^{-1}(x)) - w_{\bS'}(x)|\right) \\
&\quad + \sum_{x \in \bB^d} \sum_{y \in \Psi^{-1}(x)} w_{\bS}(y) \cdot \left((1 + \xi) \cdot \|y - c_{\phi_*(y)}\|^2 + (1 + 1/\xi) \cdot \|x - y\|^2\right) \\
&= (1 + \xi) \cdot \left(\sum_{y \in \bB^d} w_{\bS}(y) \cdot \|y - c_{\phi_*(y)}\|^2\right) \\ &\qquad+ 4 \left(\sum_{x \in \bB^d} |w_{\bS}(\Psi^{-1}(x)) - w_{\bS'}(x)|\right) \\
&\qquad + (1 + 1/\xi) \cdot \left(\sum_{y \in \bB^d} w_{\bS}(y) \cdot \|\Psi(y) - y\|^2\right) \\
&\leq (1 + \xi) \cdot \cost_{\bS}(\bC) + 4(1 + 1/\xi) \cdot \MT(\Psi, \bS, \bS'),
\end{align*}
as desired.
\end{proof}

With \Cref{lem:transport-cost-to-apx} ready, we turn our attention back to the proof of \Cref{lem:opt-transport-to-coreset}.

\begin{proof}[Proof of \Cref{lem:opt-transport-to-coreset}]
Consider any (ordered) set $\bC$ of centers. Let $\phi_*$ be the mapping that maps every point to its closest center in $\bC$. On the one hand, applying the second inequality of \Cref{lem:transport-cost-to-apx}, we get
\begin{align*}
\cost_{\bS'}(\bC) &\leq (1 + 0.5\xi) \cdot \cost_{\bS}(\bC) + 4(1 + 2/\xi) \cdot \MT(\Psi, \bS, \bS') \\
&\leq (1 + \xi) \cdot \cost_{\bS}(\bC) + 4(1 + 2/\xi)t,
\end{align*}
where the second inequality follows from the assumed upper bound on $\MT(\bS, \bS')$.

Let $\Psi: \bB^d \to \bB^d$ be such that $\MT(\Psi, \bS, \bS') = \MT(\bS, \bS')$. We can use the first inequality of \Cref{lem:transport-cost-to-apx} to derive the following inequalities.
\begin{align*}
\cost_{\bS'}(\bC) &= \cost_{\bS'}(\phi_*, \bC) \\
(\text{\Cref{lem:transport-cost-to-apx}}) &\geq \frac{1}{1 + 0.5\xi} \cdot \cost_{\bS}(\phi_* \circ \Psi) - \frac{4(1 + 2/\xi)}{1 + 0.5\xi} \cdot \MT(\Psi, \bS, \bS') \\
&\geq (1 - 0.5\xi) \cdot \cost_{\bS}(\phi_* \circ \Psi) - 4(1 + 2/\xi) \cdot \MT(\bS, \bS').
\end{align*}
As a result, we can conclude that $\bS'$ is an $(\xi, 4(1 + 2/\xi) t)$-coreset of $\bS$ as desired.
\end{proof}

\section{Frequency and Vector Summation Oracles in Local Model}
\label{app:hist-local-model}

In this section, we explain the derivations of the bounds in \Cref{subsec:histogram-prelim} in more detail. To do so, we first note that given an algorithm for histogram (resp., bucketized vector summation), we can easily derive an algorithm for generalized histogram (resp., generalized bucketized vector summation) with a small overhead in the error. This is formalized below. (Note that, for brevity, we say that an algorithm runs in time $\rt(\cdot)$ if both the encoder and the decoder run in time at most $\rt(\cdot)$.) We remark that, while we focus on the local model in this section, \Cref{lem:hist-to-genhist} works for any model; indeed we will use it for the shuffle model in the next section.

\begin{lemma} \label{lem:hist-to-genhist}
Suppose that there is a $\rt(n, Y, \eps, \delta)$-time $(\eps, \delta)$-DP $(\eta(n, \eps, \delta), \beta(n, \eps, \delta))$-accurate algorithm for histogram. Then, there is an $O(T \cdot \rt(nT, Y, \eps/T, \delta/T))$-time $(\eps, \delta)$-DP \\ $(\eta(nT, \eps/T, \delta/T), \beta(nT, \eps/T, \delta/T))$-accurate algorithm for generalized histogram.

Similarly, suppose that there is a $\rt(n, Y, d, \eps, \delta)$-time $(\eps, \delta)$-DP $(\eta(n, d, \eps, \delta), \beta(n, d, \eps, \delta))$-accurate algorithm for bucketized vector summation. Then, there is an $O(T \cdot \rt(nT, Y, d, \eps/T, \delta/T))$-time  $(\eps, \delta)$-DP $(\eta(nT, d, \eps/T, \delta/T), \beta(nT, d, \eps/T, \delta/T))$-accurate algorithm for generalized bucketized vector summation.
\end{lemma}

\begin{proof}
Suppose there is an $(\eps, \delta)$-DP $(\eta(n, \eps, \delta), \beta(n, \eps, \delta))$-accurate algorithm for histogram. To solve generalized histogram, each user runs the $T$  encoders in parallel, each on an element $y \in Y_i$ and with $(\eps/T, \delta/T)$-DP. By basic composition, this algorithm is $(\eps, \delta)$-DP. On the decoder side, it views the randomized input as inputs from $nT$ users and then runs the standard decoder for histogram. As a result, this yields an $(\eta(nT, \eps/T, \delta/T), \beta(nT, \eps/T, \delta/T))$-accurate algorithm for generalized histogram. The running time claim also follows trivially.

The argument for bucketized vector summation is analogous to the above.
\end{proof}

The above lemma allows us to henceforth only focus on histogram and bucketized vector summation.

\subsection{Histogram Frequency Oracle}

In this subsection, we briefly recall a frequency oracle of~\citet{BassilyNST20} (called \textsc{ExplicitHist} in their paper), which we will extend in the next section to handle vectors. We assume that the users and the analyzer have access to public randomness in the form of a uniformly random $Z \in \{\pm 1\}^{|Y| \times n}$. Note that while this requires many bits to specify, as noted in~\citet{BassilyNST20}, for the purpose of the bounds below, it suffices to take $Z$ that is pairwise independent in each column (but completely independent in each row) and thus it can be compactly represented in $O(n \log |Y|)$ bits.

The randomizer and analyzer from~\citet{BassilyNST20} can be stated as follows.

\begin{algorithm}[h!]
\caption{ExplicitHist Encoder}
\begin{algorithmic}[1]
\Procedure{ExplicitHistEncoder$_{\eps}(x_i; Z)$}{}
\State $\tx_i \leftarrow Z_{x_i, i}$
\State 
$
y_i = 
\begin{cases}
\tx_i & \text{ with probability } \frac{e^{\eps}}{e^{\eps} + 1} \\
-\tx_i & \text{ with probability } \frac{1}{e^{\eps} + 1}
\end{cases}
$
\State \Return $y_i$
\EndProcedure
\end{algorithmic}
\end{algorithm}

\begin{algorithm}[h!]
\caption{ExplicitHist Decoder.}
\begin{algorithmic}[1]
\Procedure{ExplicitHistDecoder$_{\eps}(v; y_1, \dots, y_n; Z)$}{}
\State \Return $\frac{e^{\eps} + 1}{e^{\eps} - 1} \cdot \sum_{i \in [n]} y_i \cdot Z_{v, i}$
\EndProcedure
\end{algorithmic}
\end{algorithm}

\citet{BassilyNST20} proved the following guarantee on the above frequency oracle:
\begin{theorem}
\textsc{ExplicitHist} is an $(O(\sqrt{n \log(|Y|/\beta)} / \eps), \beta)$-accurate $\eps$-DP algorithm for histogram in the local model. Moreover, it can be made to run in time $\poly(n, \log |Y|)$.
\end{theorem}

The above theorem in conjunction with the first part of \Cref{lem:hist-to-genhist} implies \Cref{thm:hist-local}.

\subsection{Vector Summation Oracle}

\citet{DuchiJW13} proved the following lemma\footnote{See expression (19) and Lemma 1 of~\citet{DuchiJW13}. Note that we use $L = 1$; their choice of $B = \frac{e^{\eps} + 1}{e^{\eps} - 1} \cdot \frac{\pi \sqrt{d} \Gamma\left(\frac{d - 1}{2} + 1\right)}{\Gamma\left(\frac{d}{2} + 1\right)}$ indeed satisfies $B = O(\sqrt{d})$.}, which can be used to aggregate $d$-dimensional vectors of bounded Euclidean norm.

\begin{lemma}[\citep{DuchiJW13}]
For every $d \in \N$ and $\eps \in (0, O(1))$, there exists $B = \Theta(\sqrt{d}/\eps)$ such that there is a polynomial-time $\eps$-DP algorithm $\cR^{\myvec}_{\eps}$ in the local model that, given an input vector $x$, produces another vector $z$ such that $\|z\| = B$ and $\E[z] = x$.
\end{lemma}

Notice that this algorithm allows us to compute an estimate of a sum of vectors by simply adding up the randomized vectors; this gives an error of $O(\sqrt{dn}/\eps)$. Below we combine this with the techniques of~\citet{BassilyNST20} to get a desired oracle for vector summation. Specifically, the algorithm, which we call \textsc{ExplicitHistVector}, are presented below (where $y_i$ is the bucket and $x_i$ is the vector input).

\begin{algorithm}[h!]
\caption{ExplicitHistVector Encoder}
\begin{algorithmic}[1]
\Procedure{ExplicitHistVectorEncoder$_{\eps}(y_i, x_i; Z)$}{}
\State $z_i \leftarrow \cR_{\myvec}(Z_{y_i, i} \cdot x_i)$
\State \Return $z_i$
\EndProcedure
\end{algorithmic}
\end{algorithm}

\begin{algorithm}[h!]
\caption{ExplicitHistVector Decoder.}
\begin{algorithmic}[1]
\Procedure{ExplicitHistVectorDecoder$_{\eps}(v; z_1, \dots, z_n; Z)$}{}
\State \Return $\sum_{i \in [n]} z_i \cdot Z_{v, i}$
\EndProcedure
\end{algorithmic}
\end{algorithm}

\begin{lemma} \label{lem:vecsum-local-appendix}
\textsc{ExplicitHistVector} is an $(O(\sqrt{n d \log(d|Y|/\beta)} / \eps), \beta)$-accurate $\eps$-DP algorithm for bucketized vector summation in the local model. Moreover, it can be made to run in time $\poly(nd, \log |Y|)$.
\end{lemma}

To prove \Cref{lem:vecsum-local-appendix}, we require a concentration inequality for sum of independent vectors, as stated below. It can be derived using standard techniques~\citep[see e.g.,][Corollary 7, for an even more general form of the inequality]{vector-concen}.

\begin{lemma} \label{lem:vector-chernoff}
Suppose that $u_1, \dots, u_n \in \R^d$ are random vectors such that $\E[u_i] = \bzero$ for all $i \in [n]$ and that $\|u_i\| \leq \sigma$. Then, with probability $1 - \beta$, we have $\left\|\sum_{i \in [n]} u_i\right\| \leq O\left(\sigma \sqrt{n \log(d/\beta)}\right)$.
\end{lemma}

\begin{proof}[Proof of \Cref{lem:vecsum-local-appendix}]
Since we only use the input $(x_i, y_i)$ once as the input to $\cR^{\myvec}_{\eps}$ and we know that $\cR^{\myvec}_{\eps}$ is $\eps$-DP, we  can conclude that \textsc{ExplicitHistVectorDecoder} is also $\eps$-DP.

To analyze its accuracy, consider any $v \in Y$. For $i \in [n]$, let \begin{align*}
u_i &=
\begin{cases}
z_i \cdot Z_{v, i} - x_i & \text{ if } v \in V, \\
z_i \cdot Z_{v, i} & \text{ if } v \notin Y.
\end{cases}
\end{align*}
Notice that the error of our protocol at $v$ is exactly $\sum_{i \in [n]} u_i$. Furthermore, from the guarantee of $\cR^{\myvec}_{\eps}$, it is not hard to see that $\E[u_i] = \bzero$ and that $\|u_i\| \leq \|z_i - x_i\| \leq O(\sqrt{d}/\eps)$. As a result, applying \Cref{lem:vector-chernoff}, we can conclude that with probability $1  - \beta$, $\left\|\sum_{i \in [n]} u_i\right\| \leq O(\sqrt{nd\log(d/\beta)}/\eps)$, as desired.

Similar to \textsc{ExplicitHist}, it suffices to take $Z$ that is pairwise independent in each column, which can be specified in $O(n \log |Y|)$ bits. As a result, the total running time of the algorithm is $\poly(nd, \log |Y|)$ as desired.
\end{proof}

\Cref{lem:vecsum-local-appendix} and the second part of \Cref{lem:hist-to-genhist} imply \Cref{lem:vec-sum-local}.

\section{Shuffle Model}
\label{app:shuffle-dp}

In this section, we derive our bound for shuffle DP (\Cref{thm:main-apx-shuffle}).

The shuffle DP model~\citep{bittau17,erlingsson2019amplification,CheuSUZZ19} has gained traction due to it being a middle-ground between the central and local DP models.  In the shuffle DP model, a shuffler sits between the encoder and the decoder; this shuffler randomly permutes the messages from the encoders before sending it to the decoder (aka \emph{analyst}).  
Two variants of this model have been studied in the literature: in the single-message model, each encoder can send one message to the shuffler and in the multi-message model, each encoder can send multiple messages to the shuffler.  As in the local DP model, a one-round (i.e., non-interactive) version of the shuffle model can be defined as follows. (Here we use $\cX$ to denote the set of possible inputs and $\cY$ to denote the set of possible messages.)

\begin{definition}
For an $n$-party protocol $\mathcal{P}$ with encoding function $\enc$ that produces $m$ messages per user, and for an input sequence $\bx \in \mathcal{X}^n$, we let $S_{\bx}^{\enc}$ denote the distribution on $\mathcal{Y}^{nm}$ obtained by applying $\enc$ on each element of $x$ and then randomly shuffling the outputs. 
\end{definition}


\begin{definition}[Shuffle DP]
A protocol $\mathcal{P}$ with encoder $\enc: \mathcal{X} \to \mathcal{Y}^{m}$ is $(\epsilon, \delta)$-DP in the shuffle model if the algorithm whose output distribution is $S_{\bx}^{\enc}$ is $(\epsilon, \delta)$-DP. 
\end{definition}

Recent research on the shuffle DP model includes work on aggregation~\citep{BalleBGN19, ghazi2020private, balle2020private,GKMPS21-icml}, histograms and heavy hitters~\citep{ghazi2019power, balcer2019separating, ghazi2020pure, GKMP20-icml}, and counting distinct elements \citep{balcer2021connecting, chen2020distributed}.


\subsection{Frequency and Vector Summation Oracles}

We start by providing algorithms for frequency and vector summation oracles in the shuffle DP model. These are summarized below.

\begin{theorem}[\citep{ghazi2019power}]
\label{thm:hist-shuffle}
There is an $(O(\poly\log\left(\frac{|Y| T}{\delta \beta}\right)/\eps), \beta)$-accurate $(\epsilon, \delta)$-DP algorithm for generalized histogram in the shuffle model. The encoder and the decoder run in time $\poly\left(nT/\eps,  \log\left(\frac{|Y|}{\delta \beta}\right)\right)$.
\end{theorem}

\begin{theorem} \label{thm:vec-sum-shuffle}
There is an $(O( \frac{T \sqrt{d}}{\epsilon} \cdot \poly\log(dT/(\delta \beta))), \beta)$-accurate $(\epsilon, \delta)$-DP algorithm for generalized bucketized vector summation in the shuffle model. The encoder and the decoder run in time $\poly(ndT/\beta, \log\left(\frac{|Y|}{\eps\delta}\right))$.
\end{theorem}

We will prove these two theorems in the next two subsections.

\subsubsection{Frequency Oracle}

\citep{ghazi2019power} gave a frequency oracle for histogram with the following guarantee:

\begin{theorem}[\citep{ghazi2019power}]\label{th:shuffle_fo}
There is an $(O(\poly\log\left(\frac{|Y|}{\delta \beta}\right)/\eps), \beta)$-accurate $(\epsilon, \delta)$-DP algorithm for histogram in the shuffle model. The encoder and the decoder run in time \\$\poly\left(n/\eps,  \log\left(\frac{|Y|}{\delta \beta}\right)\right)$.
\end{theorem}


We note that as stated in \citet{ghazi2019power}, the protocol underlying Theorem~\ref{th:shuffle_fo} uses $\poly(|Y| \cdot n \cdot \log(1/\beta))$ bits of public randomness. This can be exponentially reduced using the well-known fact that pairwise independence is sufficient for the Count Sketch data structure (which is the basis of the proof of~\Cref{th:shuffle_fo}).

Combining~\Cref{th:shuffle_fo} and \Cref{lem:hist-to-genhist} yields \Cref{thm:hist-shuffle}.

\subsubsection{Vector Summation Oracle}
For any two probability distributions $\mathcal{D}_1$ and $\mathcal{D}_2$, we denote by $\SD(\mathcal{D}_1, \mathcal{D}_2)$ the statistical (aka total variation) distance between $\mathcal{D}_1$ and $\mathcal{D}_2$.

We next present a bucketized vector summation oracle in the shuffle model that exhibits almost central accuracy.

\begin{theorem}\label{bucketized_vector_summation_shuffle}
There is an $(O( \frac{\sqrt{d}}{\epsilon} \cdot \poly\log(d/(\delta \beta))), \beta)$-accurate $(\epsilon, \delta)$-DP algorithm for bucketized vector summation in the shuffle model. The encoder and the decoder run in time $\poly(nd/\beta, \log\left(\frac{|Y|}{\eps\delta}\right))$.
\end{theorem}

The rest of this subsection is decidated to the proof of \Cref{bucketized_vector_summation_shuffle}; note that the theorem and the second part of \Cref{lem:hist-to-genhist} immediately imply \Cref{thm:vec-sum-shuffle}.

In order to prove Theorem~\ref{bucketized_vector_summation_shuffle}, we recall the following theorem about the analysis of a variant of the split-and-mix protocol of \citet{ishai2006cryptography}. For more context, see e.g., \citep{ghazi2020private,balle2020private} and the references therein.


We start by defining the notion of secure protocols; roughly, their transcripts are statistically indistinguishable when run on any two inputs that have the same function value.
\begin{definition}[$\sigma$-secure shuffle protocols]
Let $\nu$ be a positive real number. A one-round shuffle model protocol $\mathcal{P} = (\enc, \mathcal{A})$ is said to be \emph{$\nu$-secure} for computing a function $f:\mathcal{X}^n \to \mathbb{Z}$ if for any $\bx, \bx' \in \mathcal{X}^n$ such that $f(\bx) = f(\bx')$, we have $\SD(S_{\bx}^{\enc}, S_{\bx'}^{\enc}) \le 2^{-\nu}$.
\end{definition}

\begin{theorem}[\citep{ghazi2020private,balle2020private}]\label{th:ikos_analysis}
Let $n$ and $q$ be positive integers, and $\nu > 0$ be a real number. The split-and-mix protocol of \citet{ishai2006cryptography} (\Cref{alg:ikos_encoder}) with $n$ parties and inputs in $\mathbb{F}_q$ is $\nu$-secure for $f(\bx) = \sum_{i=1}^n x_i$ when $m \geq O(1 + \frac{\nu + \log{q}}{\log{n}})$.
\end{theorem}

We will also use the discrete Gaussian distribution from \citet{canonne2020discrete}.

\begin{definition}[\citep{canonne2020discrete}]\label{th:bernstein_ineq}
Let $\mu$ and $\sigma > 0$ be real numbers. The \emph{discrete Gaussian distribution} with location $\mu$ and scale $\sigma$, denoted by $\mathcal{N}_{\mathbb{Z}}(\mu, \sigma^2)$, is the discrete probability distribution supported on the integers and defined by
\begin{equation*}
    \Pr_{X \sim \mathcal{N}_{\mathbb{Z}}(\mu, \sigma^2)}[X= x] = \frac{e^{-(x-\mu)^2/(2 \sigma^2)}}{\sum_{y \in \mathbb{Z}} e^{-(y-\mu)^2/(2 \sigma^2)}},
\end{equation*}
for all $x \in \mathbb{Z}$.
\end{definition}

We will use the following well-known concentration inequality.
\begin{definition}[Bernstein inequality]
Let $X_1, \dots, X_n$ be independent zero-mean random variables. Suppose that $|X_i| \le M$ for all $i \in \{1, \dots, n\}$, where $M$ is a non-negative real number. Then, for any positive real number $t$, it holds that
\begin{equation*}
    \Pr\left[\sum_{i=1}^n X_i \geq t\right] \le \exp\bigg(-\frac{0.5 \cdot t^2}{\sum_{i=1}^n \E[X_i^2] + \frac{1}{3} M t}\bigg).
\end{equation*}
\end{definition}

We will also need the following tail bounds for Gaussian random variables.
\begin{proposition}[Tail bound for continuous Gaussians]\label{prop:tail_bd_continuous}
For any positive real number $x$, it holds that
\begin{equation*}
    \Pr_{X \sim \mathcal{N}(0, 1)}[X \geq x] \le \frac{\exp(-x^2/2)}{x \sqrt{2 \pi}}.
\end{equation*}
\end{proposition}

\begin{proposition}[Tail bound for discrete Gaussians; Proposition $25$ of~\citep{canonne2020discrete}]\label{prop:tail_bd_discrete}
For any positive integer $m$ and any positive real number $\sigma$, it holds that
\begin{equation*}
    \Pr_{X \sim \mathcal{N}_{\mathbb{Z}}(0, \sigma^2)}[X \geq m] \le \Pr_{X \sim \mathcal{N}(0, \sigma^2)}[X \geq m-1].
\end{equation*}
\end{proposition}

We will moreover need the following upper bound on the variance of discrete Gaussians.
\begin{proposition}[Proposition~$21$ of~\citep{canonne2020discrete}]\label{prop:var_ub_discrete_Gauss}
For any positive real number $\sigma$, it holds that
\begin{equation*}
    \Var_{X \sim \mathcal{N}_{\mathbb{Z}}(0, \sigma^2)}[X] < \sigma^2.
\end{equation*}
\end{proposition}

Finally, we will use the following theorem quantifying the differential privacy property for the discrete Gaussian mechanism.
\begin{theorem}[Theorem~$15$ of~\citep{canonne2020discrete}]\label{th:cks_apx}
Let $\sigma_1, \dots, \sigma_d$ be positive real numbers. Let $Y_1, \dots, Y_d$ be i.i.d. random variables each drawn from $\mathcal{N}_{\mathbb{Z}} (0, \sigma_j^2)$. Let $M : \mathcal{X}^n \to \mathbb{Z}^d$ be a randomized algorithm given by $M(x) = q(x) + Y$, where $Y = (Y_1, \dots, Y_d)$. Then, $M$ is $(\epsilon, \delta)$-DP if and only if for all neighboring $\bx, \bx' \in \mathcal{X}^n$, it holds that
\begin{equation*}
    \delta \geq \Pr[Z > \epsilon] - e^{\epsilon} \cdot \Pr[Z < -\epsilon],
\end{equation*}
where
\begin{equation*}
    Z := \sum_{j=1}^d \frac{ (q(x)_j - q(x')_j)^2 + 2 \cdot (q(x)_j - q(x')_j) \cdot Y_j }{2 \sigma_j^2}.
\end{equation*}
\end{theorem}

Using Theorem~\ref{th:cks_apx}, we obtain the following:

\begin{corollary}\label{cor:discrete_gaussian_vector_summation}
Let $\epsilon >0$ and $\delta \in (0, 1)$ be given real numbers. Let $\mathcal{X}$ denote the set of all vectors in $\mathbb{Z}^d$ with $\ell_2$-norm at most $C$. Define the randomized algorithm $M:\mathcal{X}^d \to \mathbb{Z}^d$ as $M(x) = q(x) + Y$, where $Y = (Y_1, \dots, Y_d)$ with $Y_1, \dots, Y_d$ being i.i.d. random variables each drawn from $\mathcal{N}_{\mathbb{Z}} (0, \sigma^2)$, and where $q(x) = \sum_{i=1}^n x_i$ is the vector sum. Then, there exists $\sigma = \frac{10 C \log(d/\delta)}{\epsilon}$ for which $M$ is $(\epsilon, \delta)$-DP.
\end{corollary}

\begin{proof}[Proof of Corollary~\ref{cor:discrete_gaussian_vector_summation}]
Let $Z$ be defined as in \Cref{th:cks_apx}. We will show that $\delta \geq \Pr[Z > \epsilon]$, which by \Cref{th:cks_apx} means that the algorithm is $(\eps, \delta)$-DP. First, note that
\begin{align*}
    \E[Z] &= \frac{1}{2 \sigma^2} \sum_{j=1}^d (q(x)_j - q(x')_j)^2 = \frac{\|q(x) - q(x')\|^2}{2 \sigma^2}.
\end{align*}
Moreover, using Proposition~\ref{prop:var_ub_discrete_Gauss}, we have that
\begin{align*}
    \Var[Z] \le \frac{1}{\sigma^2} \sum_{j=1}^d (q(x)_j - q(x')_j)^2 
    = \frac{\|q(x) - q(x')\|^2}{\sigma^2}.
\end{align*}

For each $i \in [d]$, define the event $\mathcal{E}_i$ that $|Y_i| \le M$, where
\begin{equation}\label{eq:M_setting}
M = 2 \cdot \sigma \cdot \sqrt{2 \ln(2d/\delta)}.
\end{equation}
Moreover, let $\mathcal{E} = \cap_{i=1}^d \mathcal{E}_i$. Using the fact that $Y_i \sim \mathcal{N}_{\mathbb{Z}} (0, \sigma^2)$ with $\sigma \geq 1$, and applying Propositions~\ref{prop:tail_bd_continuous} and~\ref{prop:tail_bd_discrete}, we get that $\Pr[\mathcal{E}_i] \geq 1 - \delta/(2d)$. By a union bound, we obtain
\begin{equation}\label{eq:E_event_lb_prob}
\Pr[\mathcal{E}] \geq 1-\delta/2.
\end{equation}
Henceforth, we condition on the event $\mathcal{E}$. We next argue that conditioning on $\mathcal{E}$ leaves the expectation of $Z$ unchanged, and can only decrease its variance. Namely, since both the random variable $Y_j$ and the event $\mathcal{E}_j$ are symmetric around $0$, we have that $\E[Y_j | \mathcal{E}_j] = 0$, and hence,
\begin{align}
    \E[Z | \mathcal{E}] &= \sum_{j=1}^d \frac{ (q(x)_j - q(x')_j)^2 + 2 \cdot (q(x)_j - q(x')_j) \cdot \E[Y_j | \mathcal{E}_j] }{2 \sigma_j^2} 
    = \frac{\|q(x) - q(x')\|^2}{2 \sigma^2}
    = \E[Z].  \label{eq:exp_conditional}
\end{align}

Moreover, we have that
\begin{align}
    \Var[Z | \mathcal{E}] &= \sum_{j=1}^d \frac{(q(x)_j - q(x')_j)^2}{\sigma_j^4} \cdot \Var[Y_j | \mathcal{E}_j]  \nonumber\\ 
    &= \sum_{j=1}^d \frac{(q(x)_j - q(x')_j)^2}{\sigma_j^4} \cdot \E[Y_j^2 | \mathcal{E}_j]  \nonumber\\ 
    &\le \sum_{j=1}^d \frac{(q(x)_j - q(x')_j)^2}{\sigma_j^4} \cdot \E[Y_j^2]\label{eq:cond_event_bd}\\ 
    &\le \frac{\|q(x) - q(x')\|^2}{\sigma^2},\label{eq:var_conditional}
\end{align}
where inequality~(\ref{eq:cond_event_bd}) follows from the definition of the event $\mathcal{E}_j$, and inequality~(\ref{eq:var_conditional}) follows from Proposition~\ref{prop:var_ub_discrete_Gauss}. Applying the Bernstein inequality (Theorem~\ref{th:bernstein_ineq}), we get that:
\begin{align}
    \Pr[Z > \epsilon | \mathcal{E}] &\le \exp \bigg(- \frac{0.5 \cdot (\epsilon - \E[Z | \mathcal{E}])^2}{\Var[Z | \mathcal{E}] + C \cdot M \cdot (\epsilon - \E[Z | \mathcal{E}])/(3 \sigma^2)}\bigg) \le \frac{\delta}{2},\label{eq:cond_prob_dev_ub}
\end{align}
where the last inequality follows from plugging in~(\ref{eq:M_setting}),~(\ref{eq:exp_conditional}), and~(\ref{eq:var_conditional}), and using a sufficiently large $\sigma = \frac{10 C \log(d/\delta)}{\epsilon}$. Finally, combining~(\ref{eq:cond_prob_dev_ub}) and~(\ref{eq:E_event_lb_prob}), we deduce that $\Pr[Z > \epsilon] \le \delta$.
\end{proof}

We are now ready to prove Theorem~\ref{bucketized_vector_summation_shuffle}.

\begin{proof}[Proof of Theorem~\ref{bucketized_vector_summation_shuffle}]
The pseudocode for the encoder and decoder of the protocol is given in Algorithms~\ref{alg:shuffle_buck_vec_sum_encoder} and~\ref{alg:shuffle_buck_vec_sum_decoder} respectively. Also, note that the number $t$ of incoming messages in Algorithm~\ref{alg:shuffle_buck_vec_sum_decoder} is equal to $s \cdot m$. We point out that the sum in Algorithm~\ref{alg:shuffle_buck_vec_sum_encoder} is in $\mathbb{F}_q$. We set:
\begin{itemize}
\item $\eta \gets \frac{1}{n}$.
\item $s \gets \frac{2 \cdot n}{\beta}$.
\item $\sigma \gets \frac{20 \log(s d/\delta)}{\epsilon}$.
\item $p \gets $ smallest prime larger than $\frac{2n}{\eta} + 20 \sigma \log(s d /\beta)$.
\item $m \gets O(1 + \frac{ \log{(2dp/\delta)}}{\log{n}})$.
\end{itemize}

\paragraph{Privacy Analysis.}
The analyst observes the output of the shuffler, which is the multiset $\cup_{i=1}^n \mathcal{M}_i$ of messages sent by all the users. This is equivalent to observing a vector $a \in \mathbb{Z}^s$, where $s$ is the number of buckets set in Algorithm~\ref{alg:shuffle_buck_vec_sum_encoder}. Let $\nu = \log(2d/\delta)$ and $m = O(1 + \frac{\nu + \log{p}}{\log{n}})$. Applying Theorem~\ref{th:ikos_analysis} along with a union bound, we get that the distribution of $a$ is $(d \cdot 2^{-\nu})$-close in statistical distance to the output of the central model protocol that computes the true vector sum and then adds a $\mathcal{N}_{\mathbb{Z}}(0, \sigma^2)$ noise random variable to each of the $d$ coordinates. The latter protocol is $(\epsilon, \delta/(2s))$-DP for $\sigma = \frac{20 \log(sd/\delta)}{\epsilon}$ by Corollary~\ref{cor:discrete_gaussian_vector_summation}.\footnote{A similar argument with Discrete Laplace, instead of Discrete Gaussian, noise was previously used in the proof of Lemma 4.1 of \citet{balle2020private}.} By a union bound, we thus conclude that the output of the shuffler is $(\epsilon, \delta)$-DP.

\paragraph{Utility Analysis.}
First, since we pick $s = 2n/\beta$, the probability that a collision occurs between the hash value of a given $y \in Y$ and the hash value of another bucket held by one of the users is at most $n/s \leq \beta / 2$. In other words, with probability at least $1 - \beta/2$ there is no collision with the given $y \in Y$. 

Secondly, by Propositions~\ref{prop:tail_bd_discrete} and~\ref{prop:tail_bd_continuous}, we have that with probability at least $1 - \beta/(2 s)$, it holds that for a discrete Gaussian added to each coordinate of each bucket during an execution of Algortihm~\ref{alg:shuffle_buck_vec_sum_encoder}, its absolute value is at most $10\sigma \cdot \log(sd/\beta)$. Thus, by a union bound, with probability $1 - \beta/2$, it holds that the noise to each coordinate of each bucket is at most $10\sigma \cdot \log(sd/\beta)$ in absolute value.

When the above two events hold, the error is the sum of four independent components: one due to quantization of the input vectors, and the other due to the privacy noise.

To upper-bound the quantization error, note that the error in each coordinate of each input is at most $\eta$. Thus, the error per coordinate of the sum is at most $\eta \cdot n$. Thus, the $\ell_2$-error due to quantization in the vector sum of any bucket is at most $\eta \cdot n \cdot \sqrt{d}$. This is at most $O(\sqrt{d})$ for any $\eta = O(1/n)$.

From the second event, we immediately have that the $\ell_2$-error due to the privacy noise is at most $10\sigma \cdot \log(sd/\beta) \cdot \sqrt{d}$.

Putting things together, a union bound implies that with probability at least $1 - \beta$, the total $\ell_2$-error is at most $O(\sigma \cdot \sqrt{d} \cdot \log(sd/\beta)) = O( \frac{\sqrt{d}}{\epsilon} \cdot \poly\log(d/(\delta \beta)))$.

\paragraph{Communication and Running Time.}
Each of the $n$ users sends in total $O\bigg(\frac{n}{\beta} \cdot  \log(\frac{nd}{\beta\epsilon\delta})\bigg)$ messages each consisting of at most $\log(\frac{nd}{\beta\epsilon\delta})$ bits. Note that pairwise independence can be used to reduce the number of random bits of public randomness, resulting in a total running time for the analyst which is $\poly(nd/\beta, \log\left(\frac{|Y|}{\eps\delta}\right))$.
\end{proof}

\begin{algorithm}[h!]
\caption{Encoder in Split-and-Mix Protocol}
\begin{algorithmic}[1]\label{alg:ikos_encoder}
\Procedure{SplitAndMixEncoder$(x_i, p, m)$}{}
\State Sample $z_1, \dots, z_{m-1}$ i.i.d. uniformly at random from $\mathbb{Z}_p$.
\State $z_m \gets x_i - \sum_{j=1}^{m-1} z_j$.
\State \Return the multiset $\{z_1, \dots, z_m\}$.
\EndProcedure
\end{algorithmic}
\end{algorithm}

\begin{algorithm}[h!]
\caption{Encoder in Shuffle DP Protocol for Bucketized Vector Summation}
\begin{algorithmic}[1]\label{alg:shuffle_buck_vec_sum_encoder}
\Procedure{ShuffleBVSEncoder$_{\eps, \delta}(x; Z, i, y, n, d, \sigma, \eta, s, p, m)$}{}
\State $\overline{x} \gets \lfloor x / \eta \rfloor$ (i.e., $\overline{x}_i$ is the quantization of $x_i$ up to precision $\eta$).
\State $\mathcal{M} \gets \{\}$.
\For{$(\ell, k) \in \{1, \dots, s\} \times \{1, \dots, d\}:$}
\SubIf $\ell = Z_{1, y}$:
\SubSubState $u_{\ell} \gets x_k$.
\SubElse
\SubSubState $u_{\ell} \gets 0$.
\EndIf
\SubIf $i = 1$:
\SubSubState $u_{\ell} \gets u_{\ell} + \mathcal{N}_{\mathbb{Z}}(0, \sigma^2)$. \label{step:discrete-gaussian-addition}
\EndIf
\SubState $\mathcal{M} \gets \mathcal{M} \cup (\{\ell\} \times \{k\} \times $ SplitAndMixEncoder$(u_{\ell}, p, m))$.
\EndFor
\State \Return $\mathcal{M}$
\EndProcedure
\end{algorithmic}
\end{algorithm}

\begin{algorithm}[h!]
\caption{Decoder in Shuffle DP Protocol for Bucketized Vector Summation.}
\begin{algorithmic}[1]\label{alg:shuffle_buck_vec_sum_decoder}
\Procedure{ShuffleBVSDecoder$((\ell_1, k_1, z_1), \dots, (\ell_t, k_t, z_t); Z, p, y)$}{}
\For{$j = 1, \dots, d$}
\SubState $v_j \gets \sum_{i \in \{1, \dots, t\}: \ell_{i} = Z_{1, y}, k_i = j} z_{i}$.
\SubIf $v_j > p/2:$ then $\overline{v}_j \gets (v_j-p) \eta$.
\SubElse  $~\overline{v}_j \gets v_j \eta$.
\EndIf
\EndFor
\State \Return $(v_1, \dots, v_d)$.
\EndProcedure

\end{algorithmic}
\end{algorithm}

We end by remarking that in~\Cref{alg:shuffle_buck_vec_sum_encoder} we let only the first user add the discrete Gaussian noise (Line~\ref{step:discrete-gaussian-addition}). This is only for simplicity of analysis; it is possible to split the noise between the users and achieve the same asymptotic guarantees (see~\citet{KLS21}).

\subsection{Proof of \Cref{thm:main-apx-shuffle}}

\begin{proof}[Proof of \Cref{thm:main-apx-shuffle}]
Let $\beta = 0.1$. From \Cref{thm:hist-local}, there is an $(\eta, 0.1\beta / N_T)$-accurate $0.5\eps$-DP algorithm for generalized histogram in the shuffle model with
\[
\eta = O\left(\poly\log\left(\frac{|\cL_1 \cup \cdots \cup \cL_T| \cdot T \cdot N_T}{\delta \beta}\right)/\eps\right).
\]
Since we set $T = O(\log n)$ (in \Cref{thm:net-tree-main-guarantee}), $N_T = k^{O_\alpha(1)} \cdot \poly\log n$ by our choice of parameters, and since $|\cL_1 \cup \cdots \cup \cL_T| \leq \exp(O(Td'))$ by a volume argument, we get $\eta = O(\poly\log(nd/\delta) / \eps)$.

Since we set $T = O(\log n)$ (in \Cref{thm:net-tree-main-guarantee}) and $N_T = k^{O_\alpha(1)} \cdot \poly\log n$ from our choice of parameters, we get $\eta = O(\sqrt{d} \cdot  \poly\log(nd/\delta) / \eps)$.

Similarly, from \Cref{lem:vec-sum-local}, there is an $(\teta, 0.1\beta / N_T)$-accurate $0.5\eps$-DP algorithm for generalized histogram with
\[
\teta = O\left( \frac{T \sqrt{d}}{\epsilon} \cdot \poly\log(dT/(\delta \beta))\right),
\]
which as before yields $\teta = O(\sqrt{d} \cdot  \poly\log(nd/\delta) / \eps)$.  Plugging this into \Cref{thm:apx-main}, we indeed arrive at a one-round $(\eps, \delta)$-shuffle DP $(\kappa (1 + \alpha), k^{O_{\alpha}(1)} \cdot \sqrt{d} \cdot \polylog(nd/\delta) / \eps)$-approximation algorithm for \kmeans (with failure probability $0.1$).  It is easy to verify that the encoder and the decoder run in time $\poly(n, d, k^{O_\alpha(1)}, \log(1/\delta))$.
\end{proof}

\section{Additional Experiment Details}
\label{app:exp}


\textbf{Parameter Settings.} Our experiments show that bucketized vector summation oracles often contribute to the final objective more than that of the histogram oracle; thus, we allocate more privacy budget to the former compared to the latter ($0.9\eps$ and $0.1\eps$ respectively). We view the number of levels $T$ and thresholds $\tau_1, \dots, \tau_T$ to be hyperparameters and roughly tune them. In the end, we find that $T = \lceil \log_2(k) \rceil + 3$ and only branching when the approximate frequency is at least $1.5 \lfloor n / k \rfloor$ 
give reasonably competitive objectives with little amount of tuning, and the results reported below are for these parameters. Another heuristic we find to be helpful is to split the dataset, instead of the privacy budget $\eps$, over the number of levels of the tree, i.e., randomly split the data into $T$ partitions with only the users in the $i$th partition contributing to the frequency oracle at level-$i$ of the tree. 

\paragraph{Effects of Separation of Gaussians.}
Recall that we use $r$ to denote the ratio between the separation of a of each pair of centers divided by its expected cluster size. While our plots in the main body use $r = 100$, we remark that such a large ratio is often unnecessary. Specifically, when $k = 8$, we often observe that $r \geq 8$ already gives essentially as good a result as $r = 100$; this is presented in~\Cref{fig:exp_supplement}.

\begin{figure*}[htb]
\centering
\includegraphics[width=0.4\textwidth]{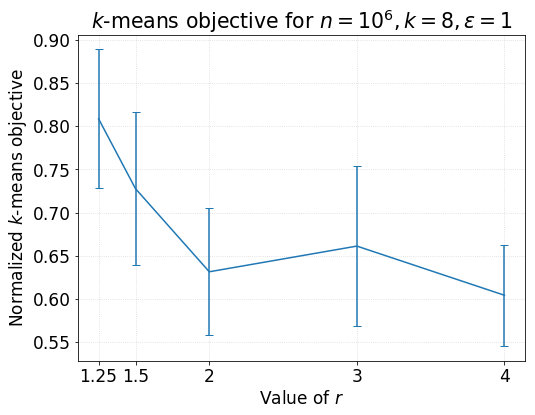}
\includegraphics[width=0.4\textwidth]{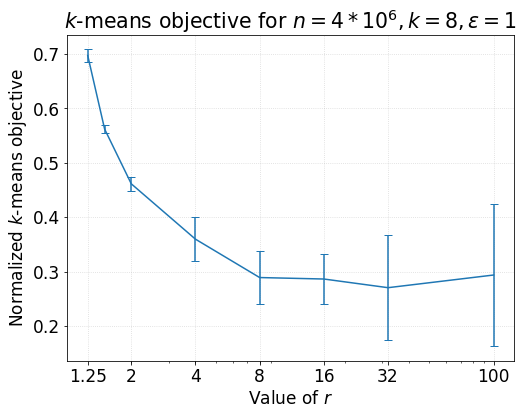}
\caption{Normalized \kmeans objective of the output clusters for varying $r$, the separation parameter of the synthetic dataset. Each set of parameters is run 10 times; the average and the standard deviation of the normalized \kmeans objectives are included.}
\label{fig:exp_supplement}
\end{figure*}


\end{document}